\title{Trust and Betrayals:\\Reputational Payoffs and Behaviors without Commitment}
\author{Harry PEI\footnote{Department of Economics, Northwestern University.
I thank Daron Acemoglu, Ricardo Alonso, Guillermo Caruana, Isa Chaves,
Mehmet Ekmekci, Jeff Ely, Drew Fudenberg, Olivier Gossner, Bruno Jullien, Aditya Kuvalekar,
Elliot Lipnowski, Debraj Ray, Bruno Strulovici, Saturo Takahashi, Jean Tirole, Juuso Toikka, Tristan Tomala, Alex Wolitzky, and Muhamet Yildiz for helpful comments.}}
\date{\today}

\documentclass[11pt]{article}
\usepackage[T1]{fontenc}

\usepackage{amsmath}
\usepackage{graphicx}
\usepackage{amsfonts}
\usepackage{amsthm}
\usepackage{setspace}
\usepackage{tikz}
\usepackage{amssymb}
\usetikzlibrary{patterns}
\usepackage{sgame}
\usepackage{color}
\usepackage{hyperref}
\usepackage{times}
\usepackage{enumitem}
\usepackage{csquotes}
\usepackage{multirow,array}
\usepackage[none]{hyphenat}
\usepackage[top=0.9in, bottom=0.8in, left=0.85in, right=0.85in]{geometry}
\begin{document}
\maketitle
\numberwithin{equation}{section}

\noindent I study a repeated game in which a patient player (e.g., a seller) wants to win the trust of some myopic opponents (e.g., buyers) but can strictly benefit from betraying them. Her benefit from betrayal is strictly positive and is her persistent private information. I characterize every type of patient player's highest equilibrium payoff. Her persistent private information affects this payoff only through the lowest benefit in the support of her opponents' prior belief. I also show that in every equilibrium which is optimal for the patient player, her on-path behavior is nonstationary, and her long-run action frequencies are pinned down for all except two types. Conceptually, my payoff-type approach incorporates a realistic concern that no type of reputation-building player is immune to reneging temptations.
Compared to commitment-type models, the incentive constraints for all types of patient player lead to a sharp characterization of her highest attainable payoff and novel predictions on her behaviors.\\


\newtheorem{Proposition}{\hskip\parindent\bf{Proposition}}[section]
\newtheorem{Theorem}{\hskip\parindent\bf{Theorem}}
\newtheorem*{Theorem1}{\hskip\parindent\bf{Theorem 1'}}
\newtheorem*{Theorem2}{\hskip\parindent\bf{Theorem 2'}}
\newtheorem*{Observation}{\hskip\parindent\bf{Observation}}
\newtheorem{Lemma}{\hskip\parindent\bf{Lemma}}[section]
\newtheorem*{Lemma2}{\hskip\parindent\bf{Lemma 3.7.2 in MS (2006)}}
\newtheorem{Corollary}{\hskip\parindent\bf{Corollary}}
\newtheorem{Definition}{\hskip\parindent\bf{Definition}}
\newtheorem{Assumption}{\hskip\parindent\bf{Assumption}}
\newtheorem{Condition}{\hskip\parindent\bf{Condition}}
\newtheorem{Claim}{\hskip\parindent\bf{Claim}}
\begin{spacing}{1.5}
\section{Introduction}
I study people's incentives to build reputations when the role models that they want to imitate are also strategic and are tempted to deviate.
For example,
consider a seller who promises to deliver high-quality products to her clients.
After the clients agree to purchase, the seller is tempted to undercut quality on aspects that are hard to verify.
Suppose all clients know for sure that the seller faces a strictly positive cost to supply high quality but does not know its exact magnitude,
can the seller benefit from reputations for having a low cost?
How does a low-cost seller's temptation to undercut quality affect the incentives to build and milk reputations?


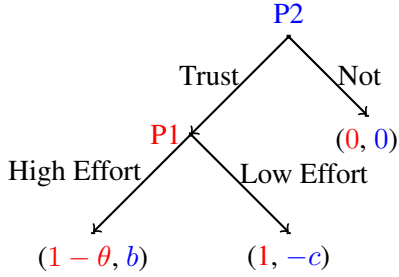
\begin{figure}
\begin{center}
\begin{tikzpicture}[scale=0.26]
\draw [->, thick] (0,10)--(-5,5);
\draw [->, thick] (0,10)--(4,6);
\draw [->, thick] (-5,5)--(-10,0);
\draw [->, thick] (-5,5)--(0,0);
\draw (0,10)--(-2,8)node[left]{Trust};
\draw (0,10)--(2,8)node[right]{Not};
\draw (-5,5)--(-7,3)node[left]{High Effort};
\draw (-5,5)--(-3,3)node[right]{Low Effort};
\draw [ultra thick] (0,9.9)--(0,10.1)node[above, blue]{P2};
\draw [ultra thick] (-5.1,5)--(-4.9,5)node[left, red]{P1};
\draw [ultra thick] (4,6)--(4,5.9)node[below]{({\color{red}{$0$}}, {\color{blue}{$0$}})};
\draw [ultra thick] (-10,0)--(-10,-0.1)node[below]{({\color{red}{$1-\theta$}}, {\color{blue}{$b$}})};
\draw [ultra thick] (0,0)--(0,-0.1)node[below]{({\color{red}{1}}, {\color{blue}{$-c$}})};
\end{tikzpicture}
\caption{The stage game, where $\theta \in (0,1)$, $b>0$, $c>0$}
\end{center}
\end{figure}

To address these questions,
I analyze a \textit{payoff-type reputation model} in which there is no commitment type, and
all types of the reputation-building player are rational and have strict incentives to misbehave. My baseline model
features a \textit{sequential-move trust game} (Figure 1, page 2) played repeatedly  between a patient player $1$ (e.g., a seller) and an infinite sequence of player $2$s (e.g., buyers), arriving one in each period and each plays the game only once.\footnote{My results are robust when players move \textit{simultaneously} in the stage game, and are robust under perturbations of player $2$'s discount factor (see section 5 and Appendix \ref{secA}). The assumption that buyers are myopic fits into applications such as durable good markets where each buyer has unit demand, and online platforms such as Uber where a buyer is unlikely to interact with the same seller twice.}
In every period, player $1$ wants to win her opponents' trust, but has a strict incentive to exert low effort once trust is granted. Her cost of exerting high effort is \textit{perfectly persistent} and is her private information,
which I call her \textit{type}. Every player $2$ observes the outcomes of all past interactions, and prefers to trust player $1$
only when he expects
high effort will be chosen
with probability above some cutoff.

My analysis suggests that all types of the patient player's temptations to misbehave
introduce additional incentive constraints, which lead to new predictions on the value of reputations
as well as novel behaviorial patterns to build and milk reputations.
My results contrast to the conclusions in
existing reputation models with \textit{commitment types} (e.g., Fudenberg and Levine 1989,1992, Benabou and Laroque 1992),
given that commitment types face no incentive constraint and mechanically play some exogenous commitment strategies.

My first result (Theorem \ref{Theorem3.1}) characterizes every type of patient player's \textit{highest equilibrium payoff},\footnote{Since players move \textit{sequentially} in the stage game, the patient player's lowest equilibrium payoff is $0$. This is also the case in models with a small probability of commitment types, and is a general feature of sequential-move stage games where future player $2$s can only observe the terminal node reached in each period, but cannot observe player $1$'s (extensive-form) stage-game  strategy. This motivates my analysis of player $1$'s \textit{highest equilibrium payoff} instead of her \textit{lowest equilibrium payoff}.} which is the product of her \textit{complete information Stackelberg payoff} and an \textit{incomplete-information multiplier}.
The latter is a sufficient statistic for the effects of incomplete information, which is strictly below one, and depends \textit{only} on the lowest possible cost in the support of player $2$s' prior belief.

By comparing every type of player $1$'s highest equilibrium payoff with that in the repeated complete information game,
my result suggests that every type of player $1$, except for the lowest-cost type, strictly benefits from incomplete information.
This observation is economically interesting since obtaining high payoff requires every high-cost type to exert low effort while receiving player $2$'s trust,\footnote{Fudenberg, Kreps and Maskin (1990)'s result implies that if in equilibrium, a type exerts high effort with positive probability every time she receives player $2$'s trust, then her equilibrium payoff cannot exceed her highest payoff under complete information, equals $1-\theta$.} and player $2$'s incentive to trust
suggests that the low-cost types need to exert high effort with high enough probability.
Therefore,
in every period where player $1$ can extract information rent, her behavior inevitably reveals information about her cost, which undermines her future informational advantage.
As her discount factor increases,
the number of periods in which she needs to extract information rent  (in order to attain a given payoff) \textit{grows without bound}. Conceptually,
how can she reveal information about her persistent type
for unboundedly many times while preserving her informational advantage?

My next two results examine
the patient player's behavior in her optimal equilibria. The motivation is to understand how she builds and milks reputations
to benefit from private information in the long run.
Theorem \ref{Theorem3.2} shows that if player $1$ has two or more types, then
in \textit{each of her optimal equilibrium}, no type mixes between high and low effort at every on-path history. This result extends to a type whose cost of exerting high effort is zero. An implication is that every type of player $1$ must have \textit{strict incentives} at some on-path histories, even when she is indifferent between high and low effort in the one-shot game. This conclusion contrasts
to the \textit{Stackelberg commitment types} in canonical reputation models that mechanically play the same mixed action in every period.

My proof develops a single-crossing argument in repeated games, which suggests that the conclusion of Theorem \ref{Theorem3.2} hinges on
the incentive constraints of \textit{all types}.
For a snapshot of the proof, suppose by way of contradiction that a type who does not have the lowest cost mixes between high and low effort at every on-path history. Then the lowest-cost type exerts high effort with probability one at every on-path history. As a result, \textit{the second-lowest cost type} separates from the lowest-cost type as soon as she exerts low effort, after which
her cost becomes the lowest one
in the support of player $2$s' posterior belief, and her continuation payoff \textit{cannot} exceed that in the repeated complete information game.
This contradicts Theorem 1, which
suggests that the second-lowest cost type obtains strictly higher payoff relative to complete information in her optimal equilibrium.

Next, suppose the lowest-cost type mixes at every on-path history.
Then the highest-cost type exerts low effort for sure, which implies that after observing low effort for a bounded number of periods,
player $2$s believe that low effort occurs with probability close to one in all future periods, after which player $2$s will stop trusting player $1$, leaving the latter a payoff close to zero. This again contradicts the optimality of equilibrium.

Theorem \ref{Theorem3.3} derives bounds on player $1$'s action frequencies that apply to \textit{all of her equilibrium best replies}. I show that first,
if player $1$'s cost is not the highest one in the support,
then the relative frequency between high and low effort \textit{cannot fall below} the ratio between their probabilities in the Stackelberg action (i.e., the \textit{critical ratio}); and
second,
if player $1$'s cost is not the lowest one in the support, then
the relative frequency between high and low effort \textit{cannot exceed} the aforementioned critical ratio.
The two bounds together pin down the action frequencies for all types of the patient player,
except for ones with the highest cost and the lowest cost.

The proof of Theorem \ref{Theorem3.3} uses a similar single-crossing argument.
A merit of my approach is that the resulting bounds on behavior
apply to \textit{all equilibrium best replies},
which have stronger testable implications compared to the ones derived
in commitment-type reputation models that apply only to the patient player's \textit{equilibrium strategy}.
This distinction is economically meaningful in dynamic signaling games where information is gradually revealed over time and the informed player's strategy inevitably involves nontrivial mixing. This is because in practice,
researchers can usually observe a sample of the dynamic interaction, corresponding to a realized path of a player's actions, but not the entire distribution over her action paths, corresponding to her equilibrium strategy.

My proof of Theorem \ref{Theorem3.1} is constructive, which illustrates how the patient player builds and milks reputations in order to
maximize
her benefit from persistent private information in the long run.
Every equilibrium I construct starts from an \textit{active learning phase} in which player $2$s trust and slowly learn about player $1$'s type. Play gradually reaches an \textit{absorbing phase} after which learning stops and no type of player $1$ can extract information rent.

In the active learning phase, the lowest-cost type mixes between high and low effort in every period.
For every high-cost type,
if player $2$'s posterior belief attaches probability close to one to the lowest-cost type,
then she exerts low effort for sure;
otherwise, she plays a mixed action, with the probability of exerting high effort being strictly lower than that of the lowest-cost type.
This arrangement enables every high-cost type to \textit{rebuild} her reputation after milking it, and reduces her reputational loss
when she extracts information rent.
The equilibria I construct have an appealing feature
that
player $1$'s reputation in the active learning phase depends \textit{only} on the number of times she has exerted high and low effort in the past, but not on other more complicated metrics.


Despite every type of player $1$ can flexibly choose
her actions in the active learning phase, her action choices affect the time with which play reaches the absorbing phase as well as her continuation payoff in the absorbing phase.
The constructed timing and continuation payoffs provide all types of player $1$ the incentives to take their equilibrium actions.
The expected duration of the active learning phase
decreases with player $1$'s lowest possible cost,
increases with her patience, and increases with player $2$'s propensity to trust player $1$.

\paragraph{Related Literature:}
This paper contributes to the literature on reputations and
rational-type repeated games.

The reputation models of Fudenberg and Levine (1989, 1992) and Gossner (2011) fix the behavior of at least one type of patient player (i.e., \textit{commitment types}), and derive payoff lower and upper bounds for a strategic-type patient player.
If the patient player's stage-game actions are \textit{statistically identified} (e.g., in simultaneous-move stage game), and there exists a commitment type who mechanically plays a mixed action,
then the strategic patient player can attain strictly higher payoffs compared to
her highest payoff in the repeated complete information game.
In addition, if there exists a commitment type who mechanically plays her Stackelberg action,
then the strategic player approximately attains
her Stackelberg payoff in all equilibria,
in which case the commitment-type approach leads to a tight characterization of the patient player's optimal equilibrium payoff.
Mailath and Samuelson (2006)'s textbook provides an excellent summary of these results.
However,
\begin{enumerate}
  \item If the patient player's stage-game actions are \textit{not statistically identified}, such as in sequential-move stage games,
  then the existing results in commitment-type models \textit{do not imply} what the patient player's highest equilibrium payoff is, nor do they imply whether the patient player can strictly benefit from incomplete information. In the trust game example,
  the commitment-type reputation results only imply that the patient player's payoff is between her \textit{minmax payoff $0$} and \textit{her Stackelberg payoff}, but there is no guarantee that this payoff upper bound is attainable, nor do these results imply whether a patient player can attain strictly greater payoff compared to her highest payoff in the repeated complete information game.
  \item Even if the patient player's actions are statistically identified, such as in simultaneous-move stage games with perfect monitoring,
 characterizing her optimal equilibrium payoff and reaching the conclusion that she can
strictly benefit from incomplete information require the presence of \textit{commitment types who play mixed actions}.
  The plausibility of such types is somewhat questionable (i.e., whether they can arise from maximizing reasonable payoff functions),
      and so are the reputation results that rely on such types.
  \item The commitment-type approach is not well-suited to answer questions related to players' behaviors. This is because first,
  rational-type patient player's behavior is sensitive to
how the commitment types behave,
and the latter are exogenously assumed instead of endogenously derived. Second, players' commitment strategies
  do not respond to changes in the payoff environment, which include, but not limited to,
 the patient player's discount factor and her initial reputation.
\end{enumerate}

I propose a complementary approach in which all types of the reputation-building player face lack-of-commitment problems.
In my payoff-type model, all types' behaviors are endogenously determined and respond to changes in the payoff environment.
This include the most efficient type (i.e., lowest-cost type) that others want to imitate.

In terms of results, I characterize every type of patient player's optimal equilibrium payoff as the value of a constrained optimization problem. My characterization applies both to sequential-move and simultaneous-move stage games.
The constraints in my program unveil how the reputational type's incentive constraints as well as the uninformed players' learning affect
a patient player's returns from reputation building.

My approach also answers novel questions related to the patient player's \textit{equilibrium behavior}, such as
how the reputational types behave when they also face lack-of-commitment problems, and how the other types behave in order to exploit those rational reputational types.
By exploiting the incentive constraints for all types of the patient player, Theorems 2 and 3 derive new behavioral
predictions on how a patient player builds and milks reputations. These results
have not been obtained in reputation models with commitment types.
The lowest-cost type's equilibrium behaviors in my payoff-type model
differ from the classic examples of Stackelberg commitment types. This suggests
that my results on behavior can help to
evaluate which of the many commitment strategies are more plausible in the sense that they can arise from maximizing reasonable payoff functions.

One concern is that my payoff-type approach cannot rule out low-payoff equilibria, for example, equilibria in which all types of patient player receive their minmax payoffs.
This is a disadvantage relative to commitment-type models when players move simultaneously in the stage-game.
It is \textit{not} a disadvantage
in sequential-move stage-games since
neither approach can rule out low-payoff equilibria, in which case the only role of reputation is to allow the patient player to attain higher payoffs relative to complete information.

Related to this paper,
Weinstein and Yildiz (2016) provide a strategic foundation for \textit{nonstationary pure-strategy commitment types} in finitely repeated games
using strategic types who maximize their discounted average payoffs.
In contrast, my paper rationalizes \textit{mixed-strategy commitment types} in infinitely repeated games using strategic types
that are not only required to maximize their discounted average payoffs, but are also required to know their own payoffs (i.e., private values), and to share the same ordinal preferences over stage-game outcomes.
Caruana and Einav (2008) endogenize players' commitment in dynamic games using switching costs. My approach differs from theirs since the patient player faces no direct cost to change her behavior over time.

My paper is related to the rational-type repeated game literature
pioneered by Aumann and Maschler (1995).
Hart (1985) characterizes the set of equilibrium payoffs in repeated games with one-sided private information and no discounting.
Shalev (1994) simplifies Hart's characterization in games with private values.
P\c{e}ski (2014) studies private-value repeated games with discounting and focuses on the case in which all players'
discount factors are arbitrarily close to $1$.
Cripps and Thomas (2003) show that
when the informed player is arbitrarily patient and the uninformed player's discount factor is bounded away from $1$,
Shalev's characterization remains to be
a \textit{necessary condition}
for being an equilibrium payoff. However, it is \textit{not sufficient} in general, which means that some payoffs in that set are not attainable in equilibrium no matter how patient the informed player is.\footnote{An exception is zero-sum games, in which Aumann and Maschler (1995)'s characterization applies
under any discount factor of the uninformed player(s).
However, the stage-game studied in this paper is not zero sum. In settings with short-lived players, the belief-free equilibrium approach in H\"{o}rner, Lovo, and Tomala (2011) is not applicable. This is because at every history in which some types of the patient player can extract information rent, the short-lived player's best reply depends on his belief about the patient player's type.}

My Theorem \ref{Theorem3.1} contributes to this literature by identifying conditions that are both \textit{necessary and sufficient} for being an equilibrium payoff.
I characterize every type of
patient informed player's highest equilibrium payoff when her opponent's discount factor
is close to or equal zero. My characterization applies to an interesting class of games, which include but not limited to product choice games and entry deterrence games.


\section{The Baseline Model}\label{sec2}
My baseline model is an infinitely repeated
trust game, which highlights the lack-of-commitment problem in economic interactions. Different from canonical reputation models with commitment types, all types of the reputation-building player are rational and have qualitatively similar payoff functions.
This captures, for instance, all types of sellers are tempted to undercut quality but can strictly benefit from buyers' purchases.

In section 5, I examine the robustness of my insights under variations of my baseline model, which include simultaneous-move stage games, the uninformed player being forward-looking, and imperfect monitoring of the informed player's actions.
Generalizations of my results beyond $2 \times 2$ trust games are stated in Appendix D.

\paragraph{Stage Game:} Consider the trust game in Figure 1 (on page 2) between an informed seller (player $1$, she) and an uninformed buyer (player $2$, he). The buyer moves first, deciding whether to trust the seller (action $T$) or not (action $N$). If he chooses $N$, then both players' payoffs are normalized to $0$. If he chooses $T$, then the seller chooses between high effort (action $H$) and low effort (action $L$). If the seller chooses $L$, then her payoff is $1$ and the buyer's payoff is $-c$. If the seller chooses $H$, then her payoff is $1-\theta$ and the buyer's payoff is $b$, where:
\begin{itemize}
  \item $b>0$ is the buyer's benefit from the seller's high effort, and $c>0$ is the buyer's loss from the seller's low effort, both of which are common knowledge among players.
  \item $\theta \in \Theta \equiv \{\theta_1,...\theta_m\} \subset (0,1)$ is the seller's cost of high effort and is her private information. Without loss of generality, I assume that $0<\theta_1 < \theta_2 < ... <\theta_m<1$.
   My results extend when $\theta_1=0$ and/or $\theta_m=1$.
\end{itemize}
The unique stage-game equilibrium outcome is $N$ and the
seller's payoff is $0$.
This is because every type of seller has a strict incentive to choose $L$ after the buyer plays $T$. If the seller can optimally commit to an action $\alpha_1 \in \Delta \{H,L\}$ before the buyer moves, then every type's optimal commitment is to
play $H$ with probability
$\gamma^* \equiv \frac{c}{b+c}$ and $L$ with probability $1-\gamma^*$.
Type $\theta_j$'s optimal commitment payoff (or \textit{Stackelberg payoff}) is
\begin{equation}\label{3.2.5}
v_j^{**} \equiv 1-\gamma^* \theta_j \textrm{ for every } j \in \{1,2,...,m\},
\end{equation}
with $\gamma^* H + (1-\gamma^*) L$ her Stackelberg action.
The above comparison between the seller's Nash equilibrium payoff and her Stackelberg payoff highlights a \textit{lack-of- commitment} problem, which is of first order importance not only in business transactions (Mailath and Samuelson 2001), but also in fiscal and monetary policies (Barro 1986, Phelan 2006) and political economy (Tirole 1996).

The rest of this article sets up a repeated version of this game in which \textit{all types of sellers have strict incentives to betray}.
I examine the extent to which a patient seller can overcome her lack-of-commitment problem, as well as different types of sellers' behaviors in seller-optimal equilibria.

\paragraph{Repeated Game:} Time is discrete, indexed by $t=0,1,2,...$. A long-lived seller interacts with an infinite sequence of buyers, arriving one in each period and each plays the game only in the period he arrives.
Both $b$ and $c$ are common knowledge among players. The seller's cost $\theta$ is perfectly persistent and is her private information.  The buyers have a full support prior belief $\pi \in \Delta (\Theta)$.

Let $y_t \in \{N,H,L\}$ be the stage-game outcome in period $t$, in which $H$ stands for the buyer chooses $T$ and the seller chooses $H$, and $L$ stands for the buyer chooses $T$ and the seller chooses $L$.
Let $h^t=\{y_{s}\}_{s=0}^{t-1} \in \mathcal{H}^t$ be a public history with $\mathcal{H} \equiv \bigcup_{t=0}^{+\infty} \mathcal{H}^t$ the set of public histories.\footnote{There is \textit{no public randomization device} in my baseline model. My results apply as long as future buyers \textit{cannot perfectly monitor the seller's mixed actions}, which is a standard assumption in the repeated game and reputation literature. This is satisfied in my baseline model, repeated simultaneous-move games with and without public randomizations, and repeated sequential-move games in which the public randomization is realized before player $2$ moves or after player $1$ moves. This assumption is violated
when players move sequentially in the stage-game and the public randomization is realized in between player $2$'s and player $1$'s moves.}
Let $A_1 \equiv \{H,L\}$ and $A_2 \equiv \{T,N\}$.
Let $\sigma_2: \mathcal{H} \rightarrow \Delta(A_2)$ be the buyers' strategy.
Let $\sigma_1 \equiv (\sigma_{\theta})_{\theta \in \Theta}$,
with type $\theta$ seller's strategy being
$\sigma_{\theta}: \mathcal{H} \rightarrow \Delta (A_1)$, which
specifies this type of seller's action choices at every public history conditional on the buyer choosing $T$.

The seller's stage-game payoff is denoted by $u_1(\theta,y_t)$, which depends on her type $\theta$ and the stage-game outcome $y_t \in \{N,H,L\}$.
The seller's discount factor is $\delta \in (0,1)$.
Type $\theta$ seller chooses $\sigma_{\theta}$ in order to maximize:
\begin{equation}\label{3.2.1}
  \mathbb{E}^{(\sigma_{\theta},\sigma_2)}\Big[  \sum_{t=0}^{\infty}(1-\delta)\delta^t u_1(\theta,y_{t}) \Big],
\end{equation}
with $\mathbb{E}^{(\sigma_{\theta},\sigma_2)}[\cdot]$ the expectation over $\mathcal{H}$ under the probability measure induced by $(\sigma_{\theta},\sigma_2)$.

\section{Results}
Theorem \ref{Theorem3.1} characterizes every type of patient seller's highest equilibrium payoff.
Theorem \ref{Theorem3.2} shows that no type of seller uses stationary strategies or completely mixed strategies in any seller-optimal equilibrium.
Theorem \ref{Theorem3.3} derives bounds on the seller's action frequencies
that apply to all of her equilibrium best replies.
These novel predictions rely on \textit{all types of seller's incentive constraints}, i.e., their incentives to betray and to receive trust.
These are not obtained in commitment-type models since commitment types face no incentive constraints.

\subsection{Patient Player's Optimal Equilibrium Payoff}
Since there are $m$ types,
the seller's  \textit{payoff} is an $m$-dimensional vector $v \equiv (v_1,...,v_m) \in \mathbb{R}^m$, in which the $j$th entry of $v$ represents the discounted average payoff of type $\theta_j$. Let
$v^* \equiv (v_1^*,...,v_m^*)$, with:
\begin{equation}\label{3.3.1}
    v_j^* \equiv \underbrace{(1-\gamma^* \theta_j)}_{\textrm{Type $\theta_j$'s Stackelberg payoff}} \underbrace{\frac{1-\theta_1}{1-\gamma^* \theta_1}}_{\textrm{incomplete-information multiplier}},  \textrm{ for every } j \in \{1,2,...,m\}.
\end{equation}
\begin{Theorem}\label{Theorem3.1}
If $\pi$ has full support, then for every $\varepsilon>0$, there exists $\underline{\delta} \in (0,1)$ such that for every $\delta \in (\underline{\delta},1)$,
\begin{itemize}
\item[1.] There exists no Bayes Nash equilibrium (BNE) in which type $\theta_1$'s payoff is strictly more than $v_1^*$.\\  There exists no BNE in which type $\theta_j$'s payoff is more than $v_j^*+\varepsilon$ for some $j \in \{2,...,m\}$.
  \item[2.] There exists a sequential equilibrium in which the seller's payoff is within $\varepsilon$ of $v^*$.\footnote{The notion of sequential equilibrium in this infinitely repeated game is introduced by P\c{e}ski (2014).}
\end{itemize}
\end{Theorem}
The proof is in Appendix A and B.
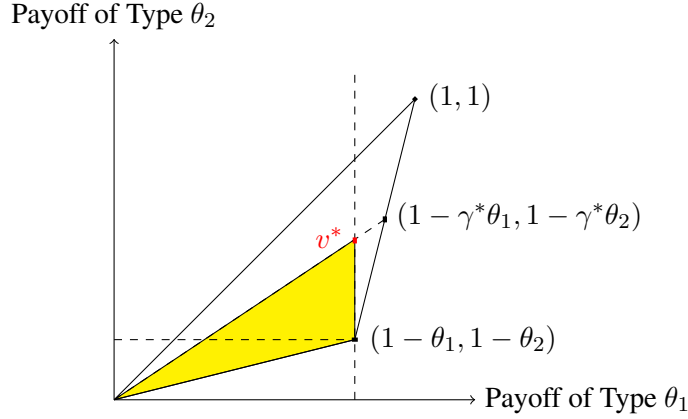
\begin{figure}\label{Figure1}
\begin{center}
\begin{tikzpicture}[scale=0.4]
\draw [fill=yellow] (0,0)--(8,2)--(8,5.33333333333333333)--(0,0);
\draw [->] (0,0)--(12,0)node[right]{Payoff of Type $\theta_1$};
\draw [->] (0,0)--(0,12)node[above]{Payoff of Type $\theta_2$};
\draw (0,0)--(10,10)--(8,2)--(0,0);
\draw [dashed] (8,0)--(8,11);
\draw [dashed] (0,0)--(9,6);
\draw [dashed] (0,2)--(8,2);
\draw [ultra thick] (7.9,2)--(8.1,2)node[right]{$(1-\theta_1,1-\theta_2)$};
\draw [ultra thick] (9.95,9.95)--(10.05,10.05)node[right]{$(1,1)$};
\draw [ultra thick] (9,5.9)--(9,6.1)node[right]{$(1-\gamma^* \theta_1,1-\gamma^* \theta_2)$};
\draw [ultra thick, red] (8,5.4)node[left]{$v^*$}--(8,5.2);
\end{tikzpicture}
\caption{When $m=2$, player $1$'s highest equilibrium payoff $v^*$ in red and her
equilibrium payoff set in yellow.}
\end{center}
\end{figure}
According to Theorem \ref{Theorem3.1}, $v_j^*$ is type $\theta_j$ patient seller's highest equilibrium payoff, and the highest equilibrium payoffs for all types can be (approximately) attained in the same equilibrium. I use different solution concepts in the two statements of Theorem 1 to strengthen my result. In particular, the necessary conditions for being an equilibrium payoff  applies under a weak solution concept (BNE), and the equilibrium that attains $v^*$ can survive demanding refinements such as sequential equilibrium.


Equation (\ref{3.3.1}) provides a tractable formula for type $\theta_j$ seller's highest equilibrium payoff,
which is the product of her \textit{complete information Stackelberg payoff} and an \textit{incomplete information multiplier}.
This multiplier is strictly below $1$ and is common for all types.
Interestingly, it depends only on the \textit{lowest possible cost} in the support of buyers' prior belief, but not on the other types in the support and the probability of each type.

To understand the intuition behind formula (\ref{3.3.1}), I state a lemma which
relates $v_j^*$ to the value of a constrained optimization problem, with proof in Appendix B.3:
\begin{Lemma}\label{L4.1}
For every $j \in \{1,2,...,m\}$, the value of the following constrained optimization problem is $v_j^*$:
\begin{equation}\label{1.1}
    \max_{\alpha \in \Delta \{N, H,L\}} \Big\{ (1-\theta_j) \underbrace{\alpha(H)}_{\textrm{probability of outcome H}} +\underbrace{\alpha(L)}_{\textrm{probability of outcome L}}\Big\},
\end{equation}
subject to:
\begin{equation}\label{1.2}
    (1-\theta_1) \alpha(H) +\alpha(L) \leq 1-\theta_1,
\end{equation}
and
\begin{equation}\label{1.3}
 \alpha(H) \geq \frac{\gamma^*}{1-\gamma^*}\alpha(L).
\end{equation}
\end{Lemma}
Next, I map the choice variable in the optimization problem $\alpha$  into the repeated incomplete information game. For given Bayes Nash equilibrium $\sigma \equiv \big(
(\sigma_{\theta})_{\theta \in \Theta},\sigma_2
\big)$, type $\theta_j$'s equilibrium payoff in the repeated game equals her expected payoff in a one-shot interaction under outcome distribution $\alpha^j \in \Delta \{N, H,L\}$, where
\begin{equation*}
 \alpha^j (y) \equiv  \mathbb{E}^{(\sigma_{\theta_j},\sigma_2)} \big[
    \sum_{t=0}^{\infty} (1-\delta)\delta^t \mathbf{1}\{y_t=y\}
    \big], \textrm{ for every } y \in \{N,H,L\}.
\end{equation*}
Replacing $\alpha$ with $\alpha^j$, the objective function (\ref{1.1}) is type $\theta_j$'s equilibrium payoff in the repeated game.
According to Lemma \ref{L4.1}, the necessity of constraints (\ref{1.2}) and (\ref{1.3}) implies that
type $\theta_j$'s equilibrium payoff \textit{cannot exceed} $v_j^*$.
I explain intuitively why constraints (\ref{1.2}) and (\ref{1.3}) are necessary for $\alpha^j$, with formal proofs in Appendix B.

Constraint (\ref{1.2}) requires that by adopting the equilibrium strategy of type $\theta_j$, type $\theta_1$ cannot receive payoff strictly higher than $1-\theta_1$, the latter is type $\theta_1$'s
highest equilibrium payoff when her type is common knowledge (Fudenberg, Kreps and Maskin 1990).
In Appendix B.1, I show by induction on the number of types that type $\theta_1$'s payoff cannot exceed $1-\theta_1$.
Intuitively, type $\theta_1$ is the most efficient type, she has no good candidate to imitate, so she cannot strictly benefit from incomplete information.
The distribution $\alpha^j$ needs to satisfy constraint (\ref{1.2}) since in equilibrium, type $\theta_1$ cannot receive strictly higher payoff by deviating to the strategy of type $\theta_j$.
Such a constraint is missing in commitment-type models since commitment types face no
incentive constraints.

Constraint (\ref{1.3}) arises from buyers' incentive constraints and their learning. This is because first, in all except for a bounded number of periods, buyers' predictions about the seller's actions are arbitrarily close to the latter's equilibrium actions
(according to the seller's true type).
Second, a buyer has no incentive to play $T$ at $h^t$ unless he expects $H$ to be played at $h^t$ with probability at least $\gamma^*$. The two parts together imply that for every $\varepsilon>0$, there exists $\underline{\delta} \in (0,1)$, such that in every equilibrium where $\delta>\underline{\delta}$,
      \begin{equation*}
\frac{\alpha^j(H)}{\alpha^j(L)}=        \frac{\mathbb{E}^{(\sigma_{\theta_j},\sigma_2)} \Big[
    \sum_{t=0}^{\infty} (1-\delta)\delta^t \mathbf{1}\{y_t=H\}
    \Big]}{\mathbb{E}^{(\sigma_{\theta_j},\sigma_2)} \Big[
    \sum_{t=0}^{\infty} (1-\delta)\delta^t \mathbf{1}\{y_t=L\}
    \Big]} \geq \frac{\gamma^*-\varepsilon}{1-\gamma^*+\varepsilon}, \textrm{ for every } j \in \{1,2,...,m\}.
      \end{equation*}
The necessity of
constraint (\ref{1.3}) is obtained by sending $\varepsilon$  to $0$.

After understanding the \textit{necessity} of these constraints, the challenging step is to establish their \textit{sufficiency}, i.e.,
there exist equilibria that approximately attain $v^*$.
This is somewhat puzzling since for every $j \geq 2$, type $\theta_j$ needs to extract information rent (i.e., playing $L$ for sure in periods where the buyer plays $T$)
for unbounded number of periods to receive discounted average payoff strictly greater than $1-\theta_j$. For this to be incentive compatible from the buyers' perspective, other types of the seller need to play $H$ with high enough probability.
As a result, a type that extracts information rent inevitably \textit{reveals information about her type}. This
undermines
her ability to extract information rent in the future since her informational advantage deteriorates every time she receives high stage-game payoff.
In section 4.1,
I explain  how to
construct equilibrium in which every high-cost type can extract information rent for unboundedly many times
while preserving her informational advantage.

Theorem \ref{Theorem3.1} has three implications.
First, aside from  the lowest-cost type $\theta_1$, every type can \textit{strictly} benefit from incomplete information, i.e., for every $j \geq 2$, type $\theta_j$'s highest equilibrium payoff in the repeated incomplete information game $v_j^*$ is strictly greater than her highest equilibrium payoff in the repeated complete information game $1-\theta_j$.
Second, the multiplier term converges to $1$ as $\theta_1$ vanishes to $0$, and every type's highest equilibrium payoff
converges to her
Stackelberg payoff.
This suggests that
although no type of seller is immune to reneging temptations, every type of patient seller can approximately attain
her optimal commitment payoff.
\begin{Corollary}\label{Cor3.3.1}
For every $\varepsilon>0$, there exist $\underline{\delta} \in (0,1)$ and $\overline{\theta}_1 >0$ such that when $\delta >\underline{\delta}$ and $\theta_1 < \overline{\theta}_1$, there exists a sequential equilibrium in which type $\theta_j$'s equilibrium payoff is no less than $v_j^{**}-\varepsilon$ for all $j \in \{1,...,m\}$.
\end{Corollary}
Third, the lowest-cost type seller \textit{fully reveals her private information}
for unboundedly many times in every equilibrium
that is approximately optimal for the seller.
Formally, let $\mathcal{H}^{(\sigma_{\theta},\sigma_2)}$ be the set of histories that occur with positive probability under $(\sigma_{\theta},\sigma_2)$. A subset of histories $\mathcal{H}'$ is called an \textit{independent set} if no pair of elements in $\mathcal{H}'$ can be ranked via the predecessor-successor relationship. The result is stated as Corollary 2.
\begin{Corollary}
For every $N \in \mathbb{N}$ and $v =(v_1,...,v_m)$ such that $v_j > 1-\theta_j$ for every $j \in \{2,...,m\}$, there exists $\underline{\delta} \in (0,1)$ such that in every BNE that attains $v$ when $\delta > \underline{\delta}$,
there exists an independent set $\mathcal{H}'$ with $|\mathcal{H}'|>N$ and $\mathcal{H}' \subset \mathcal{H}^{(\sigma_{\theta_1},\sigma_2)}$, such that player $2$'s belief attaches probability $1$ to type $\theta_1$ for every $h^t \in \mathcal{H}'$.
\end{Corollary}
This is because for each high-cost type, there exists a pure strategy in the support of her equilibrium strategy according to which her stage-game payoff exceeds $1-\theta$ only at histories where her equilibrium strategy prescribes $L$ and player $2$ plays $T$. Doing so requires the lowest-cost type to play $H$ with positive probability, after which she fully reveals her private information.

\subsection{Patient Player's Behavior in High-Payoff Equilibria}
I focus on settings in which the seller has \textit{at least two types}, i.e., $m \geq 2$. I derive properties of the patient seller's
\textit{on-path behaviors}
that apply to all equilibria that approximately attain her highest equilibrium payoff $v^* \in \mathbb{R}^m$.

First, I show that no matter how low the cost of playing $H$ is,
no type of patient seller has completely mixed equilibrium best replies, and moreover, every type's equilibrium strategy exhibits nontrivial history dependence.

For every $\sigma \equiv (\{\sigma_{\theta}\}_{\theta \in \Theta},\sigma_2)$, I say
$\sigma_{\theta}$ is \textit{stationary}
(with respect to $\sigma$)
if
it prescribes the same action at every history that
occurs with positive probability under $(\sigma_{\theta},\sigma_2)$; I say
$\sigma_{\theta}$ is \textit{completely mixed} (with respect to $\sigma$) if
it prescribes a nontrivially mixed action at every history that
occurs with positive probability under $(\sigma_{\theta},\sigma_2)$.
\begin{Theorem}\label{Theorem3.2}
When $m \geq 2$, for every small enough $\varepsilon >0$, there exists $\underline{\delta} \in (0,1)$, such that when $\delta > \underline{\delta}$, no type of player $1$
uses stationary strategies or completely mixed strategies in any BNE that attains payoff within $\varepsilon$ of $v^*$.
Moreover, no type has a completely mixed best reply in any such BNE.\footnote{Neither $m \geq 2$ nor the seller attains payoff approximately $v^*$ is redundant for the conclusion of Theorem 2. Counterexamples are available upon request.}
\end{Theorem}
As will become clear later in the proof, the conclusion of
Theorem \ref{Theorem3.2} applies more broadly. For example,
it extends to simultaneous-move stage games, as well as to a type whose cost of exerting high effort equals $0$, i.e., $\theta_1=0$.
It suggests that in every seller-optimal equilibrium, every type has \textit{strict incentives} at some on-path histories. This is the case \textit{even when} she is indifferent between exerting high and low effort in the one-shot game.

The conclusion of Theorem \ref{Theorem3.2} contrasts to the \textit{Stackelberg commitment types} in Fudenberg and Levine (1992) and Gossner (2011), who mechanically play the same mixed action in every period. To further strengthen the motivation for excluding such types in my model,
I allow
player $1$ to have \textit{arbitrary stage-game payoff functions}
in Appendix C.2.
I show that
if a type plays a nontrivially mixed action at every on-path history, then she must be
\textit{indifferent across all outcomes} in the one-shot game (Proposition \ref{PropC.1}).
In the buyer-seller application, it translates into a type of seller who
faces \textit{zero cost} of supplying high quality,
and receives \textit{zero benefit} from buyers' purchases. Therefore, the stationary Stackelberg behavior can arise in equilibrium \textit{only} under a knife-edge stage-game payoff function.

\begin{proof}[Proof of Theorem 2:]
Suppose toward a contradiction that type $\theta_j$'s best reply is to mix at every on-path history.
Then both playing $L$ at every on-path history and playing $H$ at every on-path history
are her best replies against $\sigma_2$.
Since types with lower costs enjoy comparative advantages in playing $H$ and vice versa,
every type with cost higher than $\theta_j$ plays $L$ with probability $1$ at every on-path history,
and every type with cost lower than $\theta_j$ plays $H$ with probability $1$ at every on-path history.\footnote{If we order the states and actions according to $T \succ N$, $H \succ L$ and $\theta_1 \succ \theta_2 \succ ... \succ \theta_m$, then the stage-game payoff satisfies a monotone-supermodularity condition in Liu and Pei (2020). This is sufficient to guarantee the monotonicity of the sender's strategy with respect to her type in one-shot signalling games.
I use an implication of their result on repeated signalling games (Pei 2019).}

In what follows, I show by contradiction that none of these
pure stationary strategies are compatible with the requirement that
type $\theta_j$'s payoff is strictly above $1-\theta_j$ for every $j \geq 2$.

First, suppose there exists a type $\theta_k$ that plays $L$ with probability $1$ at every on-path history. According to Fudenberg and Levine (1989), if the seller is of type $\theta_k$, then the buyers
will eventually believe that $L$ will be played with probability greater than $1-\gamma^*$ in all future periods, after which they will have a strict incentive to play $N$. This implies that type $\theta_k$ seller's discounted average payoff is close to $0$ when $\delta$ is close to $1$.

Next, suppose $j \geq 2$ and types $\theta_1$ to $\theta_{j-1}$ play
$H$ with probability $1$ at every on-path history. Then after type $\theta_j$ plays $L$ for one period, she becomes the lowest-cost type in the support of
buyers' posterior belief. I show in Proposition B.1 of Appendix B that type $\theta_j$'s payoff in the continuation game is at most $1-\theta_j$, and therefore, her discounted average payoff cannot exceed $(1-\delta)+\delta (1-\theta_j)$.
The latter is strictly lower than $v_j^*$ when $\delta \rightarrow 1$ given that $v_j^*>1-\theta_j$ for every $j \geq 2$.

The next step is to rule out stationary strategies. Stationary (non-trivially) mixed strategies have been ruled out since I have shown that no type of seller plays completely mixed strategies in equilibrium. Stationary pure strategies have been ruled out by the above proof given that it  has established that neither playing $H$ for sure at every on-path history and playing $L$ for sure at every on-path history can be equilibrium strategies
of any type of the patient seller in equilibria that approximately attains $v^*$.
\end{proof}

Focusing on seller-optimal equilibria, my next result derives lower and upper bounds on the seller's action frequencies that apply not only to all of her equilibrium strategies, but also to all of her \textit{equilibrium best replies}.
\begin{Theorem}\label{Theorem3.3}
For every small enough $\varepsilon>0$, there exists $\underline{\delta} \in (0,1)$, such that
for every $\delta > \underline{\delta}$,
in every BNE $\sigma \equiv \big((\sigma_{\theta})_{\theta \in \Theta},\sigma_2\big)$
that attains payoff within $\varepsilon$ of $v^*$,
\begin{enumerate}
  \item For every $\theta \neq \theta_m$ and for every best reply $\widehat{\sigma}_{\theta}$ of type $\theta$'s against $\sigma_2$:
\begin{equation}\label{3.3.4}
    \frac{\mathbb{E}^{(\widehat{\sigma}_{\theta},\sigma_2)} \Big[\sum_{t=0}^{\infty} (1-\delta)\delta^t \mathbf{1}\{y_t=H\}\Big] }{ \mathbb{E}^{(\widehat{\sigma}_{\theta},\sigma_2)} \Big[\sum_{t=0}^{\infty} (1-\delta)\delta^t \mathbf{1}\{y_t=L\}\Big]}
    \geq \frac{\gamma^*-\varepsilon}{1-(\gamma^*-\varepsilon)}.
\end{equation}
  \item For every $\theta \neq \theta_1$ and for every best reply $\widehat{\sigma}_{\theta}$ of type $\theta$'s against $\sigma_2$:
\begin{equation}\label{3.3.5}
    \frac{\mathbb{E}^{(\widehat{\sigma}_{\theta},\sigma_2)} \Big[\sum_{t=0}^{\infty} (1-\delta)\delta^t \mathbf{1}\{y_t=H\}\Big] }{ \mathbb{E}^{(\widehat{\sigma}_{\theta},\sigma_2)} \Big[\sum_{t=0}^{\infty} (1-\delta)\delta^t \mathbf{1}\{y_t=L\}\Big]}
    \leq \frac{\gamma^*+\varepsilon}{1-(\gamma^*+\varepsilon)}.
\end{equation}
\end{enumerate}
\end{Theorem}

The proof is in Appendix \ref{secC.1}. The two bounds in Theorem \ref{Theorem3.3} pin down the action frequencies for all types of sellers except for the lowest-cost and the highest-cost types. Since
type $\theta_j$'s payoff is approximately $v_j^*$, it also
pins down the discounted average frequency of every stage-game outcome under every equilibrium best reply:
\begin{Corollary}
For every small enough $\varepsilon>0$, there exists $\underline{\delta} \in (0,1)$ such that
for every $\delta > \underline{\delta}$,
in every BNE $\sigma \equiv \big((\sigma_{\theta})_{\theta \in \Theta},\sigma_2\big)$
that attains payoff within $\varepsilon$ of $v^*$, for every $\widehat{\sigma}_{\theta}$ that is type $\theta \in \{\theta_2,...,\theta_{m-1}\}$'s best reply against $\sigma_2$, we have:
\begin{equation}\label{3.3.6}
   \mathbb{E}^{(\widehat{\sigma}_{\theta},\sigma_2)} \big[\sum_{t=0}^{\infty} (1-\delta)\delta^t \mathbf{1}\{y_t=H\}\big]
   \in \Big( \gamma^* \frac{1-\theta_1}{1-\gamma^* \theta_1} -\varepsilon, \textrm{ }\gamma^* \frac{1-\theta_1}{1-\gamma^* \theta_1} +\varepsilon \Big),
\end{equation}
\begin{equation}\label{3.3.7}
    \mathbb{E}^{(\widehat{\sigma}_{\theta},\sigma_2)} \big[\sum_{t=0}^{\infty} (1-\delta)\delta^t \mathbf{1}\{y_t=L\}\big]
    \in \Big( (1-\gamma^*) \frac{1-\theta_1}{1-\gamma^* \theta_1} - \varepsilon, \textrm{ } (1-\gamma^*) \frac{1-\theta_1}{1-\gamma^* \theta_1} +\varepsilon \Big),
\end{equation}
and
\begin{equation}\label{3.3.8}
\mathbb{E}^{(\widehat{\sigma}_{\theta},\sigma_2)} \big[\sum_{t=0}^{\infty} (1-\delta)\delta^t \mathbf{1}\{y_t=N\}\big]
\in \Big( \frac{(1-\gamma^*) \theta_1}{1-\gamma^* \theta_1}- \varepsilon,  \textrm{ } \frac{(1-\gamma^*) \theta_1}{1-\gamma^* \theta_1} +\varepsilon \Big).
\end{equation}
\end{Corollary}
Theorem \ref{Theorem3.3} and Corollary 3 differ from the following implication of Fudenberg and Levine (1992)'s result, that
according to every type of seller's \textit{equilibrium strategy},
the ratio between the discounted average frequency of outcome $H$ and the discounted average frequency of outcome $L$
is no less than $\gamma^*/(1-\gamma^*)$.

The reason is that
Fudenberg and Levine (1992)'s result only applies to the seller's \textit{equilibrium strategies}, while the bounds in Theorem 3 apply to
all of her \textit{equilibrium best replies}, which include but not limited to her equilibrium strategies.
As an illustrative example, the strategy of playing the Stackelberg action at every on-path history satisfies the requirement in Fudenberg and Levine (1992), but has been ruled out by Theorem \ref{Theorem3.3}. This is because for some pure strategies in the support of the seller's equilibrium strategy (such as playing $L$ at every history), the relative frequency of outcome $L$ exceeds (\ref{3.3.4}). For other pure strategies in its support (such as playing $H$ at every history), the relative frequency of outcome $L$ falls below (\ref{3.3.5}).

Conceptually,
the bounds that apply to \textit{all equilibrium best replies} have stronger testable implications compared to the ones that apply only to \textit{equilibrium strategies}. This distinction is economically important given that in many markets, researchers can only observe one or a few \textit{realized paths} of the game's outcomes instead of an entire distribution over outcome paths. The predictions of Theorem 3 and Corollary 3 apply to every pure-strategy best reply, and therefore, can be tested by observing a realized path of equilibrium play.
For example, Corollary 3 implies that the frequencies of outcome $H$ and outcome $L$ are strictly decreasing in $\theta_1$, and the frequency of outcome $N$ is strictly increasing in $\theta_1$. By observing the frequencies of outcomes along a realized path of play,
a researcher can identify whether there are more efficient types  under the buyers' belief, and if yes, he can infer the value of $\theta_1$ based on his estimation of $\gamma^*$ and the observed frequencies of different outcomes.

\section{Constructing High-Payoff Equilibria}
Let $v^N \equiv (0,0,...,0)$, $v^H \equiv (1-\theta_1,1-\theta_2,...,1-\theta_m)$, and $v^L \equiv (1,1,...,1)$, which
are the seller's payoffs from outcomes $N$, $H$, and $L$, respectively.
For every $\gamma \in [\gamma^*,1]$, let
\begin{equation}\label{4.12}
v(\gamma) \equiv   \frac{\theta_1 (1-\gamma)}{1-\gamma \theta_1} v^N
    + \frac{(1-\theta_1) \gamma}{1-\gamma \theta_1} v^H
    + \frac{(1-\theta_1) (1-\gamma)}{1-\gamma \theta_1} v^L \in \mathbb{R}^m.
\end{equation}
I depict $v(\gamma)$ in Figure 3.
One can verify that $v(\cdot)$ is a continuous function of $\gamma$ with $v(\gamma^*)=v^*$ and $v(1)=v^H$.
Statement  2 of Theorem 1 is implied by Proposition 4.1, which is shown in Appendix A:
\begin{figure}\label{Figure2}
\begin{center}
\begin{tikzpicture}[scale=0.34]
\draw [fill=yellow] (0,0)--(8,2)--(8,5.33333333333333333)--(0,0);
\draw [->] (0,0)--(12,0)node[right]{Payoff of Type $\theta_1$};
\draw [->] (0,0)--(0,12)node[above]{Payoff of Type $\theta_2$};
\draw (0,0)--(10,10)--(8,2)--(0,0);
\draw [dashed] (8,0)--(8,11);
\draw [dashed] (0,0)--(9,6);
\draw [dashed] (0,2)--(8,2);
\draw [ultra thick] (9.95,9.95)--(10.05,10.05)node[right]{$(1,1)$};
\draw [ultra thick] (9,5.9)--(9,6.1)node[right]{$(1-\gamma^* \theta_1,1-\gamma^* \theta_2)$};
\draw [blue, ultra thick] (7.9,4)--(8.1,4);
\draw [blue, dashed, ->] (8.1,4)--(10,4)node[right]{$v (\gamma)$};
\draw [ultra thick] (8,5.4)node[left]{$v^*$}--(8,5.2);
\draw [ultra thick] (7.9,2)--(8.1,2)node[right]{$(1-\theta_1,1-\theta_2)$};
\end{tikzpicture}
\caption{In an example with two types, the set of attainable payoffs in yellow and $v(\gamma)$ in blue.}
\end{center}
\end{figure}
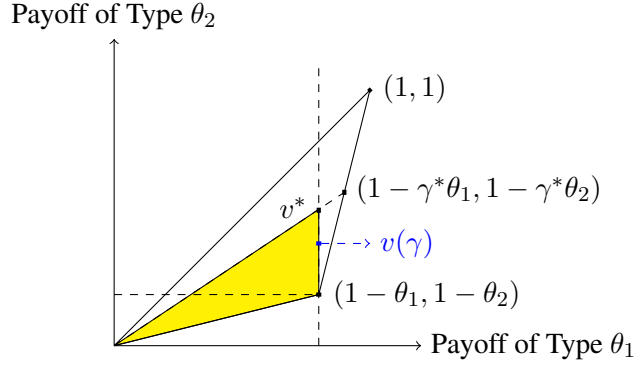
\begin{Proposition}\label{Prop3.4.1}
For every $\eta \in (0,1)$
and $\gamma \in (\gamma^*,1)$, there exists $\underline{\delta} \in (0,1)$, such that when $\delta > \underline{\delta}$ and the prior belief $\pi$ attaches probability  more than $\eta$ to type $\theta_1$, there exists equilibrium that attains $v(\gamma)$.
\end{Proposition}
In section 4.1,
I provide an overview of the equilibrium construction in an environment with \textit{two types}.
In section 4.2,
I summarize the main ideas behind the construction and
discuss the key features of players' equilibrium strategies.
In section 4.3,
I compare the equilibrium dynamics with other reputation models.

\subsection{Equilibrium Construction with Two Types: Strategies and Continuation Values}
I specify players' actions and the evolution of player $1$'s continuation value at \textit{on-path histories}.
At every \textit{off-path history}, player $2$ plays $N$, and all types of player $1$ play $L$ if player $2$ plays $T$.
I keep track of two state variables:
\begin{enumerate}
  \item $\eta(h^t) \in [0,1]$, which is the probability of type $\theta_1$ under player $2$'s belief at $h^t$,
  and is called player $1$'s \textit{reputation}, with $\eta(h^0)$ equals the prior probability of type $\theta_1$.
  \item $v(h^t) \in \mathbb{R}^2$, which is
  P1's continuation value at $h^t$, with $v(h^0) \equiv v(\gamma)$.
I verify in Appendix A.4.1 that for every $h^t$,
    $v(h^t)$
  is a convex combination of $v^N$, $v^H$ and $v^L$, i.e., $v(h^t)=p^N(h^t) v^N +p^H(h^t) v^H +p^L(h^t) v^L$.
  Hence, keeping track of $v(h^t)$ is equivalent to keeping track of
   $p^N(h^t)$, $p^H(h^t)$, and $p^L(h^t)$.
\end{enumerate}
I partition the set of histories into \textit{three classes}, depending on the value of $p^L(h^t)$:
\begin{itemize}
  \item \textbf{Class 1 Histories:} $p^L(h^t) \geq 1-\delta$.
  \item \textbf{Class 2 Histories:} $p^L(h^t) \in (0, 1-\delta)$.
  \item \textbf{Class 3 Histories:} $p^L(h^t) =0$.
\end{itemize}
Play starts from a Class 1 history, and reaches a Class 3 history in finite time, after which it stays there forever.
Active learning about player 1's type happens at Class 1 and Class 2 histories (i.e., call it an \textit{active learning phase}), but stops at Class 3 histories (i.e., call it an \textit{absorbing phase}).

\paragraph{Class 1 Histories:} Play starts from a Class 1 history where $p^L(h^t) \geq 1-\delta$. If $h^t$ belongs to Class 1, then
\begin{itemize}
  \item Player $2$ plays $T$ for sure.
  \item Player $2$'s posterior beliefs, $\eta(h^t,H)$ and $\eta(h^t,L)$, are functions of $\eta(h^t)$:
 \begin{equation}\label{4.15}
 \eta (h^t,H) = \eta^*+ \min \Big\{1-\eta^*, \big(1+\lambda (1-\gamma^*) \big) \big(\eta (h^t)-\eta^* \big) \Big\},
\end{equation}
\begin{equation}\label{4.16}
\textrm{and} \quad \eta (h^t,L)  = \eta^* + (1-\lambda \gamma^*) (\eta (h^t)-\eta^*).
 \end{equation}
\begin{itemize}
  \item $\eta^* \in (\gamma^* \eta(h^0),\eta(h^0))$ is a constant.
  Given that $\eta(h^0)>\eta^*$,
  one can verify by induction that $\eta(h^t)>\eta^*$ for every Class 1 (and Class 2) history $h^t$.
  \item $\lambda >0$ is a constant that measures the speed of player $2$'s learning.
I require
  \begin{equation}
    \big(1-\lambda \gamma^*\big)^{1-\gamma} \big(1+\lambda (1-\gamma^*)\big)^{\gamma} > 1.
\end{equation}
Since $\gamma > \gamma^*$, Taylor's theorem implies that
(4.4) is satisfied when $\lambda$ is sufficiently small.
\end{itemize}
\end{itemize}
Equations (\ref{4.15}) and (\ref{4.16}) pin down both types of player $1$'s actions at $h^t$. This is because according to Bayes Rule,
the probability with which type $\theta_1$ plays $H$ at $h^t$ is:
\begin{equation}\label{4.13}
 \frac{\eta(h^t)-\eta(h^t,L)}{\eta(h^t,H)-\eta(h^t,L)} \cdot \frac{\eta(h^t,H)}{\eta(h^t)},
\end{equation}
and the probability with which type $\theta_2$ plays $H$ at $h^t$ is:
\begin{equation}\label{4.14}
    \frac{\eta(h^t)-\eta(h^t,L)}{\eta(h^t,H)-\eta(h^t,L)} \cdot \frac{1-\eta(h^t,H)}{1-\eta(h^t)}.
\end{equation}
Plugging (\ref{4.15}) and (\ref{4.16}) into
 (\ref{4.13}) and (\ref{4.14}), every type's action at $h^t$ can be written as a function of $\eta(h^t)$.

Let $p_H(h^t)$ be
the probability of $H$ being played at $h^t$ according to player $2$'s belief at $h^t$.
Since player $2$'s belief is a martingale, $p_H(h^t) \eta (h^t,H) +(1-p_H(h^t)) \eta (h^t,L)=\eta(h^t)$.
Equation (\ref{4.16}) and $\eta(h^t)>\eta^*$ imply that  $\eta(h^t,L) \neq \eta(h^t)$. This together with the martingale condition implies that:
\begin{equation*}
 \frac{1-p_H(h^t)}{p_H(h^t)}=   \frac{\eta (h^t,H)-\eta(h^t)}{\eta(h^t)-\eta (h^t,L)}.
\end{equation*}
Plugging
(\ref{4.15})
and (\ref{4.16}) into the above equation, we have:
\begin{equation}\label{4.4}
  \frac{1-p_H(h^t)}{p_H(h^t)}=  \frac{\eta (h^t,H)-\eta(h^t)}{\eta(h^t)-\eta (h^t,L)} \leq \frac{1-\gamma^*}{\gamma^*}.
\end{equation}
This implies that $p_H(h^t) \geq \gamma^*$, which suggests that player $2$
has an incentive to play $T$ at $h^t$.

Player 1's continuation value after playing $L$ at $h^t$ is:
  \begin{equation}\label{4.17}
v(h^t,L)= \frac{p^N(h^t)}{\delta} v^N +  \frac{p^L(h^t)-(1-\delta)}{\delta} v^L + \frac{p^H(h^t)}{\delta} v^H.
\end{equation}
If $h^t$ is such that $\eta(h^t,H)<1$, then player $1$'s continuation value after playing $H$ at $h^t$ is given by:
\begin{equation}\label{4.18}
v(h^t,H)=\frac{p^N(h^t)}{\delta} v^N  + \frac{p^L(h^t)}{\delta}v^L+ \frac{p^H(h^t)-(1-\delta)}{\delta} v^H.
\end{equation}
If $h^t$ is such that $\eta(h^t,H)=1$, then player $1$'s continuation value after playing $H$ at $h^t$ is given by:
\begin{equation}\label{4.20}
  v(h^{t},H) = \frac{v_{1}(h^{t},H)}{1-\theta_1} v^H +\Big(1-\frac{v_{1}(h^{t},H)}{1-\theta_1}\Big)v^N,
\end{equation}
\begin{equation}\label{4.19}
 \textrm{with }  v_1(h^{t},H) \equiv \frac{v_1(h^{t})-(1-\delta)(1-\theta_1)}{\delta} \textrm{ and } v_1(h^t) \textrm{ is the first entry of } v(h^t).
\end{equation}
If $h^t$ is such that $\eta(h^t,H)<1$, then (\ref{4.17}) and (\ref{4.18}) imply that
both types of player $1$ are indifferent between $H$ and $L$.
If $h^t$ is such that $\eta(h^t,H)=1$, then (\ref{4.17}), (\ref{4.20}) and (\ref{4.19}) imply that
type $\theta_1$ is indifferent while type $\theta_2$ strictly prefers $L$.
Player $1$'s incentive constraints at $h^t$ are satisfied since
when $\eta(h^t,H)<1$,
both types are required to mix; and when $\eta(h^t,H)=1$, type $\theta_1$ is required to mix while type $\theta_2$ plays $L$ for sure.

\paragraph{Class 2 Histories:} At every Class 2 history $h^t$, i.e., one in which
$p^L(h^t) \in (0,1-\delta)$,
\begin{itemize}
  \item Player 2 plays $T$ for sure.
  \item Type $\theta_1$ plays $H$ for sure. Type $\theta_2$
plays $L$ with probability
$\min\{1, \frac{1-\gamma^*}{1-\eta(h^t)}\}$. Therefore, the probability with which $L$ being played at $h^t$ is
$(1-\eta(h^t)) \min\{1, \frac{1-\gamma^*}{1-\eta(h^t)}\}$, which is no more than $1-\gamma^*$.
This implies that player $2$ has an incentive to play $T$.
\end{itemize}
After player $1$ plays $L$ at $h^t$, $\eta(h^t,L)=0$ given that type $\theta_1$ plays $H$ for sure, and
player $1$'s continuation value is given by:
  \begin{equation}\label{4.23}
  v\big(h^t,L \big) \equiv   \frac{Q(h^t)}{\delta} v^H + \frac{\delta-Q(h^t)}{\delta} v^N,
  \end{equation}
  where
  \begin{equation}\label{4.24}
    Q(h^t) \equiv p^H(h^t) - \frac{1-\delta-p^L(h^t)}{1-\theta_2}
  \end{equation}
Player $1$'s continuation value after playing $H$ depends on whether $\eta(h^t,H)$ equals $1$ or not, the latter can be computed according to Bayes Rule
via $\eta(h^t)$ and different types' mixing probabilities specified above.
\begin{enumerate}
  \item If $\eta(h^t,H)<1$, then player $1$'s continuation value after playing $H$ at $h^t$ is given by (\ref{4.18}).
  \item If $\eta(h^t,H)=1$, then player $1$'s continuation value after playing $H$ at $h^t$ is given by (\ref{4.20}).
\end{enumerate}
If $h^t$ is such that $\eta(h^t,H)<1$, then (\ref{4.23}) and (\ref{4.18}) imply that
type $\theta_2$ is indifferent and type $\theta_1$ strictly prefers to play $H$.
If $h^t$ is such that $\eta(h^t,H)=1$, then (\ref{4.23}) and (\ref{4.20}) imply that
type $\theta_2$ strictly prefers to play $L$ and type $\theta_1$ strictly prefers to play $H$.
Player $1$'s incentive constraints at $h^t$ are satisfied since type $\theta_1$ is required to play $H$, while type $\theta_2$ is required to mix only if
$\eta(h^t,H)<1$, and is required to play $L$ if $\eta(h^t,H)=1$.

\paragraph{Class 3 Histories:}
If $h^t$ is such that $p^L(h^t)=0$, then
all types of player $1$ play the same action at $h^t$ so learning about her type stops. Moreover,
her continuation value at every subsequent history
is also a convex combination of $v^H$ and $v^N$, i.e., all subsequent histories belong to Class 3.
The construction of equilibrium play after reaching any Class 3 history uses Lemma 3.7.2 in Mailath and Samuelson (2006, page 99):
\begin{Lemma2}
For all $\varepsilon>0$, there exists $\underline{\delta} \in (0,1)$ such that for every $\delta \in (\underline{\delta},1)$ and every $v\in \mathbb{R}^m$ that is a convex combination of $v(1)$, $v(2)$,...,$v(k)$, there exists $\{v^s\}_{s=0}^{\infty}$ with $v^s \in \{v(1),...,v(k)\}$
such that (1) $v=\sum_{s=0}^{\infty} (1-\delta)\delta^s v^s$, and (2) for every $l \in \mathbb{N}$,
$\sum_{s=l}^{\infty} (1-\delta)\delta^{s-l} v^s$ is within $\varepsilon$ of $v$.
\end{Lemma2}
Since $v(h^t)$ is a convex combination of $v^H$ and $v^N$, the above lemma implies that when $\delta$ is above some cutoff, there exists an infinite sequence
$\{v^s\}_{t=0}^{\infty}$ with $v^s \in \{v^N,v^H\}$ such that:
\begin{itemize}
  \item[1.] $v(h^t)=(1-\delta) \sum_{s=0}^{\infty} \delta^s v^s$,
  \item[2.] For every $l \in \mathbb{N}$,
$(1-\delta) \sum_{s=l}^{\infty} \delta^{s-l} v^s$ is within $\varepsilon$ of $v(h^t)$.
\end{itemize}
The continuation play following $h^t$ is:
\begin{itemize}
  \item For every $s \in \mathbb{N}$ such that $v^s=v^H$,
  P2 plays $T$ and all types of P1 play $H$ in period $t+s$.
  \item For every $s \in \mathbb{N}$ such that $v^s=v^N$,
    P2 plays $N$ and all types of P1 play $L$ in period $t+s$.
\end{itemize}
Player $2$'s incentive constraints at Class 3 histories are trivially satisfied.
To verify player $1$'s incentive constraints, I show in Lemma A.5 of Appendix A
that $p^H(h^t)$ is bounded from below by some strictly positive number for every $h^t$ belonging to Class 1 or Class 2. This implies that
if  $h^t$ belongs to Class 3 but none of its predecessors belong to Class 3, then
 $p^H(h^t)$ is also bounded from below by a positive number. Pick $\varepsilon$ in the above lemma to be small enough, one can ensure that player $1$'s continuation value at every on-path history succeeding $h^t$ is strictly bounded away from $0$.
This implies that
 at every on-path history where player $1$ is asked to play $H$, all types of patient player $1$ have strict incentives to comply. This is because every type's continuation payoff equals $0$ if she does not comply, and her continuation payoff is strictly bounded away from $0$ if she complies.

\paragraph{Promise Keeping:} I have verified player $2$s' incentive constraints, and
have constructed continuation values under which player $1$'s incentive constraints are satisfied. What remains to be verified is the promise keeping constraint, that the continuation play at every on-path history delivers all types of player $1$ their respective continuation values.
This is established by
showing that under the above strategy profile,
 play reaches a Class 3 history in finite time with probability $1$ (implied by Lemmas A.1 and A.4 in Appendix A),\footnote{Nevertheless, the expected duration of the active learning phase (Class 1 and 2 histories) grows without bound as $\delta$ approaches $1$.}
 and player $1$'s continuation value at Class 3 histories can be delivered via a sequence of payoffs consisting only of $v^H$ and $v^N$.

\paragraph{Remark:} For illustration, I verify the conclusions of Theorems 2 and 3 in the constructed equilibrium.
For Theorem 2,
type $\theta_1$ mixes at Class 1 histories but has strict incentives at Class 2 and Class 3 histories; type $\theta_2$ mixes in the active learning phase when her reputation is low, but has a strict incentive to play $L$ when her reputation is high. She also has a strict incentive in the absorbing phase. Verifying the conclusion of Theorem 3 requires a careful examination of player $1$'s continuation values, in particular, the convex weights of $v^H$, $v^N$ and $v^L$. These three state variables keep track of the discounted average frequency of each stage-game outcome that has occurred before. In the constructed equilibrium that attains $v(\gamma)$,
\begin{itemize}
  \item Consider the discounted frequency of outcomes when type $\theta_1$ plays one of her pure-strategy best replies.
  If the posterior belief $\eta(h^t)$ never reaches $1$, then the discounted average frequency of outcome $H$ is $\gamma \frac{1-\theta_1}{1-\gamma \theta_1}$, the discounted average frequency of outcome $L$ is $(1-\gamma) \frac{1-\theta_1}{1-\gamma \theta_1}$. 
  If the posterior belief $\eta(h^t)$ reaches $1$ for some $t \in \mathbb{N}$, then the discounted average frequency of outcome $H$ is strictly higher than $\gamma \frac{1-\theta_1}{1-\gamma \theta_1}$, and the discounted average frequency of outcome $L$ is strictly lower than $(1-\gamma) \frac{1-\theta_1}{1-\gamma \theta_1}$.
  \item Consider the discounted frequency of outcomes when type $\theta_2$ plays one of her pure-strategy best replies.
  If play was at
   a Class 1 history in the period before entering the absorbing phase,
  then the discounted average frequency of outcome $H$ is $\gamma \frac{1-\theta_1}{1-\gamma \theta_1}$, the discounted average frequency of outcome $L$ is $(1-\gamma) \frac{1-\theta_1}{1-\gamma \theta_1}$. 
  If play was at a Class 2 history in the period before entering the absorbing phase, then the discounted average frequency of outcome $H$ is strictly lower than $\gamma \frac{1-\theta_1}{1-\gamma \theta_1}$, and the discounted average frequency of outcome $L$ is strictly higher than $(1-\gamma) \frac{1-\theta_1}{1-\gamma \theta_1}$.
\end{itemize}
The conclusion of Theorem \ref{Theorem3.3} is verified once
taking $\gamma$ to be arbitrarily close to $\gamma^*$.

\subsection{Summary of Ideas Behind the Construction}
To start with, learning about player $1$'s type is indispensable for her to attain payoff $v^*$, or anything above $v^H$.
According to Corollary 2,
player $1$ needs to be able to \textit{rebuild} her reputation after playing $L$.\footnote{This is true except for histories such that $p^L(h^t) \leq 1-\delta$, which happens only if $L$ has been played too frequently before.} This is because otherwise,
she
cannot extract information rent for unboundedly many periods and as a result, cannot obtain discounted average payoff significantly above $1-\theta$.
To make reputation rebuilding feasible, both types of player $1$ mix between $H$ and $L$, except for histories where $\eta(h^t)$ is close to $1$, in which case the low-cost type $\theta_1$ mixes between $H$ and $L$ while the high-cost type $\theta_2$ plays $L$ for sure. Intuitively,
this arrangement reduces the high-cost type's reputation loss when she shirks for sure and extracts information rent.\footnote{Histories at which type $\theta_2$ has a strict incentive to play $L$ is indispensable for her to attain payoff $v_2^*$. This is because otherwise, playing $H$ at every history where $T$ is played with positive probability is one of type $\theta_2$'s equilibrium best reply, and under this strategy of hers, her payoff in every period cannot exceed $1-\theta_2$, which is strictly lower than $v_2^*$.}

The presence of the absorbing phase (Class 3 histories) is to provide all types of player $1$ the incentives to mix in the active learning phase (Class 1 \& 2 histories). Despite player $1$ can flexibly choose her actions when active learning takes place, her action choices affect her continuation payoff after reaching the absorbing phase, as well as the calendar time at which play enters the absorbing phase.
For example, if player $1$ plays $L$ too frequently in the beginning, then $p^L(h^t)$ decreases more quickly and play reaches a Class 3 history as soon as $p^L(h^t)=0$, after which player $1$'s continuation value is low and can no longer extract information rent.

A more subtle situation occurs when player $1$ plays $H$ too frequently, which decreases $p^H(h^t)$ and increases $p^L(h^t)$. It raises a concern that  $p^H(h^t)/p^L(h^t)$ may fall below $\frac{\gamma^*}{1-\gamma^*}$, after which player $1$'s continuation payoff cannot be delivered in any equilibrium. To address this, play enters the absorbing phase when $\eta(h^t)$ reaches $1$, after which
the continuation play consists only of outcomes $N$ and $H$.
Type $\theta_1$ is indifferent between entering the absorbing phase (by playing $H$) and remaining in the active learning phase (by playing $L$), while type $\theta_2$ strictly prefers the latter option. This is because type $\theta_1$ has a comparative advantage in playing $H$.
Requirement (4.4) on $\lambda$  implies that given $p^H(h^0)/p^L(h^0)$ is strictly greater than $\frac{\gamma^*}{1-\gamma^*}$ and play remains in the active learning phase at $h^t$, the ratio of the convex weights $p^H(h^t)/p^L(h^t)$ is also strictly greater than $\frac{\gamma^*}{1-\gamma^*}$. This is formally stated as Lemma A.1 in Appendix A, which ensures that the equilibrium play eventually reaches a Class 3 history.

What needs to be considered next is the speed of learning $\lambda$. In order to maximize player $1$'s equilibrium payoff, one needs to maximize the frequency of outcome $L$ while simultaneously providing incentives for player $2$s to trust.
This leads to the role of \textit{slow learning}, i.e., $\lambda$ being low.
To understand why, first,
player $2$'s incentive to trust translates into an upper bound on the \textit{relative rate of learning}, given by (4.7).
In a nutshell, it requires
that the magnitude of reputation improvement after playing $H$ to be small enough relative to the magnitude of reputation deterioration after playing $L$.
Second, fixing the relative rate of learning and the long-run frequencies of outcomes $H$ and $L$,
the amount of reputation loss per period vanishes as the \textit{absolute speed of learning} goes to zero.
As a result, lowering the absolute speed of learning improves
player $1$'s long-term reputation without compromising on  player $2$s' willingness to
trust. This allows for an increase in the long-run frequency of outcome $L$ without sacrificing player $1$'s reputation, which helps to improve player $1$'s payoff.
\subsection{Comparing Equilibrium Dynamics}\label{sub4.3}
I compare the equilibrium dynamics in my model to those in models with behavioral biases (Jehiel and Samuelson 2012), reputation cycles (Sobel 1985, Phelan 2006, Liu 2011, Liu and Skryzpacz 2014), gradual learning (Benabou and Laroque 1992, Wiseman 2005, 2012,
Ekmekci 2011),  mixed-strategy commitment types (Mathevet, Pearce and Stacchetti 2019),
capital-theoretic models of reputations (Board and Meyer-ter-Vehn 2013, Bohren 2018, Dilm\'{e} 2018), and models of reputation sustainability (Cripps, Mailath and Samuelson 2004).

\paragraph{Analogical-Based Reasoning Equilibria:} The patient player alternates between her actions to manipulate her opponents' belief is reminiscent of the \textit{analogy-based reasoning equilibria} in Jehiel and Samuelson (2012).
In their model, there are multiple commitment types who are playing stationary mixed strategies, and one strategic type who can flexibly choose her actions. The short-run players mistakenly believe that the strategic type is playing a stationary strategy. In the trust game, their results imply that the strategic long-run player's behavior experiences a \textit{reputation building phase} in which she plays $H$ for a bounded number of periods, followed by a \textit{reputation manipulation phase} that resembles the active learning phase in my model where she alternates between $H$ and $L$ according to the Stackelberg frequencies.
The short-run players' posterior belief fluctuates within a small neighborhood of the cutoff belief, implying that
the long-run player's type is never fully revealed.

Comparing my model to theirs,
there are two qualitative differences in the reputation dynamics
that highlight the distinctions between rational and analogical-based short-run players.
First, learning stops in finite time in my model while it lasts forever in theirs.
This is driven by the
constraint that type $\theta_1$'s equilibrium payoff cannot exceed $1-\theta_1$, which comes from
the rational short-run players' ability to correctly predict the long-run player's average action in \textit{every period}.
This constraint is absent
when short-run players use analogy-based reasoning
since they can only correctly predict \textit{the long-run player's average action across all periods}.
Second, the short-run players learn the true state with positive probability in every high-payoff equilibrium of my model,
while in Jehiel and Samuelson (2012), the probability with which they learn the true state is zero.
This is because analogy-based short-run players' posterior beliefs depend only on the empirical frequencies of the observed actions. That is to say, their beliefs are not responsive enough to each individual observation.

\paragraph{Reputation Building-Milking Cycles:} The behavioral pattern that a patient player builds her reputation in order to milk it in the future has been identified in commitment-type models with either changing types (Phelan 2006), or limited memories (Liu 2011, Liu and Skrzypacz 2014).
In terms of differences, first, the reputation cycles in Phelan (2006), Liu (2011) and Liu and Skrzypacz (2014) can last forever while
learning stops in finite time in mine.
This is driven by the constraint that the lowest-cost type's equilibrium payoff cannot exceed $1-\theta_1$, which
arises because this type is rational and has strict incentives to misbehave.

Second, reputations are built and milked \textit{gradually} in my model while in theirs, the agent's reputation falls to its lower bound every time she milks it. This is because the commitment types in their models never betray. As a result, one misbehavior reveals the long-run player's rationality.
In my model, the behaviors of good and bad types are close
since all types share the same ordinal preferences over stage-game outcomes and face strict temptations to misbehave. In the long-run player's optimal equilibrium, the lowest-cost type misbehaves with positive probability for unbounded number of periods, which does not reduce her own payoff while at the same time, covering up for the other types when they milk reputations.

This feature of gradual learning is supported empirically by several studies of online markets. As documented in Dellarocas (2006) and Bar-Isaac and Tadelis (2008), consumers judge the quality of sellers based on their reputation scores, which are usually obtained via averaging the ratings they obtained in the past. Empirical works have documented that one recent negative rating neither significantly affects the amount of sales nor the prices of a reputable seller who has obtained many positive ratings in the past, which matches the
 dynamics in my model.

\paragraph{Reputation Models with Gradual Learning:} Benabou and Laroque (1992) and  Ekmekci (2011) study reputation games with commitment types and the long-run player's actions are imperfectly monitored.
In their equilibria, learning also happens gradually since the short-run players cannot tell the difference between intended cheating and exogenous noise. In contrast, my model has perfect monitoring of stage-game outcomes but has no commitment type. Gradual learning occurs since the reputational type cheats with positive probability.
The different driving forces behind gradual learning also lead to different long-run outcomes.
In my model, reputation building-milking cycles stop in finite time while in theirs, reputation cycles last forever.
As mentioned before, this is driven by the rational reputational type's strict incentive to misbehave, which implies an upper bound on her equilibrium payoff.
Another difference is that the short-run players never fully learn the long-run player's type
in Benabou and Laroque (1992) and  Ekmekci (2011), while in mine, the lowest-cost type fully reveals her private information for unboundedly many times
 in every high-payoff equilibrium (Corollary 2).

The class of equilibria I construct start from  a phase where active learning takes place followed by an absorbing phase where learning stops.
Different from the learning phases in Wiseman (2005, 2012), during which players experiment and learn about their payoffs,
the learning phase in my model is constructed so that the patient player can extract information rent and can attain strictly higher payoff compared to the complete information benchmark.
As a result, the length of the learning phase in my model increases with the discount factor, while it does not vary with the discount factor in Wiseman (2005, 2012).

\paragraph{Mixed-Strategy Commitment Types:}
Mathevet, Pearce and Stacchetti (2019) construct patient-player-optimal equilibria in a repeated communication game in which the stage-game follows from the leading example of Kamenica and Gentzkow (2011). They examine a commitment-type model in which with positive probability,
the sender is committed and sends
 messages according to her optimal disclosure policy in every period.

A qualitative difference between their equilibrium and mine emerges \textit{when the long-run player's reputation is low}, i.e., after she has cheated too much in the past. In their equilibrium, the normal type sender can still extract information rent in the future, she never reveals her rationality, and learning about her type does not stop at low reputations. In my equilibrium, the high-cost type cannot extract information rent anymore in the future, her private information can be perfectly revealed, and learning about her type stops at low reputations.

This difference is driven by the distinction between rational reputational type in my model and non-strategic committed type in theirs. In my rational type model, if player $1$ has shirked too much in the past and has a low reputation,
then she cannot be allowed to extract information rent any more in the future.
This is because otherwise, even the \textit{lowest-cost type} will have \textit{a strict incentive to shirk}, making it no longer valuable for other types to build a reputation for being the lowest-cost type, and the reputational equilibrium unravels.

In contrast, the \textit{mixed-strategy commitment type} in their model faces no incentive constraint,
plays a mixed action in every period.  As a consequence, the
rational long-run player never fully reveals her type,
and can extract information rent in the future
no matter how much she has cheated in the past.

\paragraph{Capital-Theoretic Reputation Models:} Reputation cycles also occur in the Poisson good news models of Board and Meyer-ter-Vehn (2013) and Dilm\'{e} (2018).
They characterize Markov equilibria in which
the long-run player exerts effort when her reputation is above some cutoff.
Different from my model, reputation jumps up
immediately after the arrival of good news.
Moreover, the long-run player's reputation depends only on the most recent time of news arrival in their models while it depends on
the history of her actions in mine.

These distinctions are caused by the difference sources of learning.
In my model, learning arises from the differences in different types' behaviors, while in their models,
all types adopt the same behavior but face different news arrival rates. In terms of the applications, my model fits into online platforms where feedback arrives frequently while their Poisson models fit into markets with infrequent news arrival.

In the bad news model of Board and Meyer-ter-Vehn (2013) and the Brownian model of Bohren (2018), the informed player's effort increases in her reputation. This differs from the active learning phase in
the equilibria I construct
where the high-cost types
shirk with probability one when her reputation is close to one.

\paragraph{Sustainability of Reputations:} Cripps, Mailath and Samuelson (2004) study models with commitment types.
They show
if the monitoring structure satisfies a full support condition and  player $1$'s actions are statistically identified, then in every equilibrium, player $2$s eventually learn player $1$'s type as $t \rightarrow \infty$, and the limiting equilibrium play converges to an
equilibrium of the game where  player $1$'s type is common knowledge.

In my baseline model, the uninformed players perfectly observe the outcome in each period, which means that the full support and identification conditions are not satisfied.
This leads to qualitative differences in equilibrium outcomes. For example, player $2$s \textit{may not} learn player $1$'s type as $t \rightarrow \infty$.
Such incomplete learning occurs in my constructed equilibrium when play reaches a Class 3 history directly from a Class 1 history,
after which learning about player $1$'s stops and
player $2$'s belief attaches strictly positive probability to multiple types.\footnote{Incomplete learning of player $1$'s type \textit{does not contradict} my previous claim that in all except for a bounded number of periods, player $2$'s prediction about player $1$'s future actions is arbitrarily close to player $1$'s action under the true state. This can happen when different types of player $1$ use the same equilibrium strategy, in which case player $2$s cannot learn player $1$'s  type by observing player $1$'s actions but can precisely predict player $1$'s action in the true state.}

Moreover, my analysis needs to take into account player $1$'s payoff and behavior in \textit{finite time}, as well as the speed with which player $2$s learn. This contrasts to Cripps, Mailath and Samuelson (2004)'s
results that focus exclusively on player $2$'s belief and the equilibrium play as $t \rightarrow \infty$.
To elaborate,
in every equilibrium in which type $\theta_j$ approximately attains $v_j^*$ for some $j \geq 2$,
the expected length of the active learning phase grows without bound
as $\delta \rightarrow 1$. This is because
type $\theta_j$'s continuation payoff cannot exceed $1-\theta_j$ after learning stops,
so she can only obtain payoff
\textit{strictly higher} than $1-\theta_j$ when active learning about her type takes place.
Despite player $1$'s continuation value eventually converges to a complete information game payoff as $t \rightarrow \infty$, the time it takes for such convergence to occur is crucial for a patient player $1$'s equilibrium payoff.
\section{Concluding Remarks}\label{sec5}
I conclude by discussing the robustness of my results to simultaneous-move stage games, forward-looking buyers, and imperfect monitoring of the seller's actions (section \ref{sub5.1}). Then I list some alternative applications of my model and results (section \ref{sub5.2}).
Generalizations beyond $2 \times 2$ trust games are stated in Appendix \ref{secA}.

\subsection{Robustness of Results}\label{sub5.1}
\paragraph{Simultaneous-Move Stage Game:} Consider a simultaneous-move trust game with stage-game payoffs:
\begin{center}
\begin{tabular}{| c | c | c |}
  \hline
  - & $T$ & $N$ \\
  \hline
  $H$ & $1-\theta,b$ & $-d(\theta),0$ \\
  \hline
  $L$ & $1,-c$ & $0,0$ \\
  \hline
\end{tabular}
\end{center}
where $b,c>0$, $\theta \in \Theta \equiv \{\theta_1,\theta_2,...,\theta_m\}\subset (0,1)$ is player $1$'s persistent private information, and $d(\theta) \geq 0$  measures player $1$'s loss when she exerts high effort while player $2$ does not trust.
In the repeated version of this game, players' past action choices are perfectly monitored and the public history $h^t \equiv \{a_{1,s},a_{2,s}\}_{s=0}^{t-1}$ consists of both players' past action choices.
Other features of the game remain the same as in the baseline model.

For the results on equilibrium payoffs, recall the definition of $v^*$ in (\ref{3.3.1}).
A construction similar to the one in section 4 implies that $v^*$ is approximately attainable when $\delta$ is close to $1$.
Under a supermodularity condition on the stage-game payoffs:
\begin{equation}\label{5.1}
    0 \leq d(\theta_j)-d(\theta_i) \leq \theta_j-\theta_i \quad \textrm{for every} \quad j<i,
\end{equation}
one can show that for every $j \in \{1,2,...,m\}$, $v_j^*$ is type $\theta_j$ patient long-run player's highest equilibrium payoff.\footnote{In a repeated incomplete information game with all types of player $1$ sharing the same ordinal preferences over stage-game outcomes, the lowest equilibrium payoff is $0$. This is because repeating the outcome of $(L,N)$ is always an equilibrium of the repeated game, no matter whether the stage-game is simultaneous-move or sequential-move.}

Under the supermodularity condition in (\ref{5.1}), the conclusions in Theorem \ref{Theorem3.2} and Theorem \ref{Theorem3.3} extend
to the simultaneous-move
stage game. In particular, no type of the long-run player uses a stationary strategy or has a completely best reply in any equilibrium that approximately attains $v^*$. For the bounds on the long-run player's action frequencies that apply to all of her pure-strategy best replies, one needs to replace $y_t=H$ and $y_t=L$ in (\ref{3.3.4}) and (\ref{3.3.5}) with $a_{1,t}=H$ and $a_{1,t}=L$, respectively.

\paragraph{Stage Game with Imperfect Monitoring:} Player $1$ is an agent, for example a worker, a supplier or a private contractor. In every period, a principal (player $2$, for example an employer or a final good producer) is randomly matched with the agent. The principal then decides whether to incur a fixed cost and interact with the agent or to skip the interaction.
The agent chooses her effort from a closed interval unbeknownst to the principal.
The probability with which the service quality being high
increases with her effort. In line with the literature on incomplete contracts,
the service quality is not contractible but is observable to the agent and all the subsequent principals.
The cost of effort is linear and the marginal cost of effort is the agent's persistent private information.\footnote{Chassang (2010) studies a game in which players face similar incentives. The main difference is that the agent's cost of effort is common knowledge but the set of actions that are available in each period is the agent's private information. Tirole (1996) uses a similar model to study the collective reputation of commercial firms and that of bureaucrats.}

Players move sequentially in the stage game. Different from the baseline model, after player $2$ chooses to trust, player $1$ chooses among a continuum of effort levels $e \in [0,1]$.
The quality of the output being produced is denoted by $z \in \{G,B\}$, which is \textit{good} (or $z=G$) with probability $e$
and is \textit{bad} (or $z=B$) with complementary probability.
The cost of effort for type $\theta_i$ is $\theta_i e$. Player $1$'s benefit from her opponent's trust is normalized to $1$. Therefore, her stage-game payoff under outcome $N$ is $0$ and that under outcome $(T,e)$ is $1-\theta_i e$. Player $2$'s payoff is $0$ if he chooses $N$. His benefit from good output is $b$ while his loss from bad output is $c$, with $b,c >0$. Therefore, player $2$ is willing to trust only when player $1$'s expected effort exceeds
$\gamma^* \equiv \frac{c}{b+c}$.

Consider the repeated version of this game in which the public history consists of player $2$'s actions and the realized output quality, i.e., player $1$'s effort choice is her private information. In period $t$,
let $a_{1,t}$ be player $1$'s action, let $a_{2,t}$ be player $2$'s action, and let $z_t$ be the realized output quality. Let $h^t=\{a_{2,s},z_s\}_{s=0}^{t-1} \in \mathcal{H}^t$ be a public history  with $\mathcal{H} \equiv \bigcup_{t=0}^{+\infty} \mathcal{H}^t$ the set of public histories.
Let $h_1^t=\{a_{1,s}, a_{2,s},z_s\}_{s=0}^{t-1} \in \mathcal{H}_1^t$ be player $1$'s private history with
$\mathcal{H}_1 \equiv \bigcup_{t=0}^{+\infty} \mathcal{H}_1^t$ the set of private histories.
Let $\sigma_2 : \mathcal{H} \rightarrow \Delta (A_2)$ be player $2$'s strategy and
let $\sigma_{\theta}: \mathcal{H}_1 \rightarrow \Delta (A_1)$ be type $\theta$ player $1$'s strategy, with $\sigma_1 \equiv (\sigma_{\theta})_{\theta \in \Theta}$.

The above game with a continuum of effort, linear effort cost, and imperfect monitoring is equivalent to the baseline model with binary effort and perfect monitoring. To see this, choosing effort level $e$ under imperfect monitoring is equivalent to choosing a mixed action $eH+(1-e)L$ under perfect monitoring.
In terms of the results on payoffs, one can show that $v_j^*$ is type $\theta_j$'s highest equilibrium payoff when she is patient, and payoff vector $v^*$ is approximately attainable when $\delta$ is close to $1$. In terms of behaviors,
the bounds on the relative frequencies can be applied to realized paths of public signals, namely,
one needs to replace $y_t=H$ and $y_t=L$ in (\ref{3.3.4}) and (\ref{3.3.5}) with $(a_{2,t},z_t)=(T,G)$ and $(a_{2,t},z_{t})=(T,B)$, respectively.

\paragraph{Forward-Looking Buyer:} My results are robust when the seller faces a single buyer whose discount factor $\delta_2$ is strictly positive but close to $0$. To begin with, the constructed equilibrium that approximately attains $v^*$ remains to be an equilibrium under any $\delta_2$. This is because at every off-path history, the buyer plays $N$ and all types of player $1$ play $L$, in which case the buyer receives his minmax payoff. Hence, the buyer's strategy in the constructed equilibrium maximizes his stage-game payoff
while it cannot lower his continuation payoff.

The necessity of constraint (\ref{1.2}) relies on the observation that at every on-path history, the buyer has no incentive to play $T$ unless he expects $H$ to be played with positive probability. This remains valid when $\delta_2<\gamma^*$.
Suppose toward a contradiction that at some on-path history $h^t$,
all types of seller play $L$ for sure, but
the buyer plays $T$ with strictly positive probability.
The buyer's discounted average payoff by playing $T$ at $h^t$ is at most:
      \begin{equation*}
        \underbrace{(1-\delta_2) (-c)}_{\textrm{P2's stage-game payoff if he plays $T$ while P1 plays $L$ for sure}}+\underbrace{\delta_2 b}_{\textrm{P2's maximal continuation payoff after playing $T$}}.
      \end{equation*}
Since $\delta_2<\gamma^* \equiv \frac{c}{b+c}$, the above expression is strictly less than $0$.
This contradicts the buyer's incentive to play $T$ at $h^t$
since he can secure payoff $0$ by playing $N$ in every subsequent period.

In addition, when $\delta_2$ is close to $0$, the buyer has no incentive to play $T$ at $h^t$ unless he expects $H$ to be played with probability more than $\gamma^*-\varepsilon$, with $\varepsilon$ vanishes to $0$ as $\delta_2 \rightarrow 0$. This implies
an \textit{approximate version} of constraint (\ref{1.3}) when
the seller's discount factor $\delta_1$ is close enough to $1$:
\begin{equation}
    \alpha^j(H) \geq \frac{\gamma^*-\varepsilon}{1-\gamma^*+\varepsilon}\alpha^j(L),
\end{equation}
with
\begin{equation*}
 \alpha^j (y) \equiv  \mathbb{E}^{(\sigma_{\theta_j},\sigma_2)} \big[
    \sum_{t=0}^{\infty} (1-\delta_1)\delta_1^t \mathbf{1}\{y_t=y\}
    \big] \textrm{ for every } y \in \{N,H,L\}.
\end{equation*}
Replacing (\ref{1.3}) with constraint (5.2), the value of the constrained optimization problem is close to $v_j^*$, which converges to $v_j^*$ as $\delta_2 \rightarrow 0$. This implies the robustness of Theorem 1 to perturbations of $\delta_2$.
Given that the proofs of Theorems 2 and 3 do not use buyers' incentive constraints aside from the conclusion that $v^*$ is a patient seller's highest equilibrium payoff, those results are also robust to small perturbations of $\delta_2$.

\subsection{Alternative Applications}\label{sub5.2}
\paragraph{Capital Taxation:} Player $1$ is a government and player $2$s are a sequence of foreign investors. The stage-game is depicted in Figure 4, where $\theta \in \{\theta_1,...,\theta_m\}$ is the government's private information that measures its benefit from expropriating investors via high tax rates. I assume that $0<\theta_1<...,<\theta_m$, namely, all types of government strictly benefit from expropriation.
\begin{figure}
\begin{center}
\begin{tikzpicture}[scale=0.22]
\draw [->, thick] (0,10)--(-5,5);
\draw [->, thick] (0,10)--(4,6);
\draw [->, thick] (-5,5)--(-10,0);
\draw [->, thick] (-5,5)--(0,0);
\draw (0,10)--(-2,8)node[left]{invest};
\draw (0,10)--(2,8)node[right]{not invest};
\draw (-5,5)--(-7,3)node[left]{low tax};
\draw (-5,5)--(-3,3)node[right]{high tax};
\draw [ultra thick] (0,9.9)--(0,10.1)node[above, blue]{P2};
\draw [ultra thick] (-5.1,5)--(-4.9,5)node[left, red]{P1};
\draw [ultra thick] (4,6)--(4,5.9)node[below]{({\color{red}{$0$}}, {\color{blue}{$0$}})};
\draw [ultra thick] (-10,0)--(-10,-0.1)node[below]{({\color{red}{$1$}}, {\color{blue}{$b$}})};
\draw [ultra thick] (0,0)--(0,-0.1)node[below]{({\color{red}{$1+\theta$}}, {\color{blue}{$-c$}})};
\end{tikzpicture}
\caption{Capital Taxation Game Between Government and Investors}
\end{center}
\end{figure}
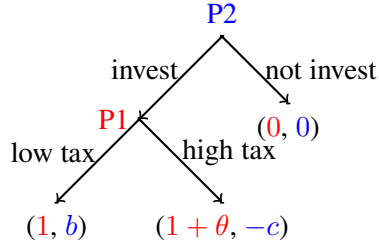
This game can be analyzed using similar techniques as my baseline model. Let $\gamma^* \equiv \frac{c}{b+c}$. One can show that type $\theta$ government's highest equilibrium payoff $v_j^*$ is the value of:
\begin{equation}\label{5.2}
    \max_{\alpha \in \Delta \{\textrm{not invest}, \textrm{high tax}, \textrm{low tax}\}} \big\{ \alpha(\textrm{low tax}) +(1+\theta_j)\alpha(\textrm{high tax}) \big\},
\end{equation}
subject to:
\begin{equation}\label{5.3}
     \alpha(\textrm{low tax}) +(1+\theta_1)\alpha(\textrm{high tax}) \leq 1,
\end{equation}
and
\begin{equation}\label{5.4}
 \alpha(\textrm{low tax}) \geq \frac{\gamma^*}{1-\gamma^*}\alpha(\textrm{high tax}).
\end{equation}
Solving this problem, one can obtain
\begin{equation}
v_j^* =\frac{1+\theta_j-\gamma^* \theta_j}{1+\theta_1-\gamma^* \theta_1}.
\end{equation}
Similar to the baseline model, $v_j^*$
depends only on $\theta_j$ and the lowest possible temptation to expropriate $\theta_1$, but not on the other aspects of incomplete information.
In addition, the lowest benefit type $\theta_1$ cannot strictly benefit from incomplete information, while types $\theta_2$ to $\theta_m$ can receive payoff strictly higher than $1$.
The results on the government's on-path behaviors in equilibria that approximately attain $v^* \equiv (v_1^*,...,v_m^*) $ extend as well.

\paragraph{Entry Deterrence/Limit Pricing Game} Player $1$ is an incumbent choosing between a low price (or \textit{fight}) and a normal price (or \textit{accommodate}). Player $2$ is an entrant deciding
whether to enter or not.
 Players' payoffs are:
      \begin{center}
\begin{tabular}{| c | c | c |}
  \hline
  - & Out & Enter \\
  \hline
  Low Price & $1-\theta,0$ & $-d(\theta),-b$ \\
  \hline
  Normal Price & $1,0$ & $0,c$ \\
  \hline
\end{tabular}
\end{center}
where $\theta$ and $d(\theta)$ are the incumbent's costs from lowering prices, which is interpreted as limit pricing if the entrant stays out, and is interpreted as predation if the entrant enters. As argued in Milgrom and Roberts (1982), $\theta$ depends on the efficiency of the incumbent's production technology, which tends to be its persistent private information.
This maps into the simultaneous-move version of the stage game once we replace Low Price with $H$, Normal Price with $L$, Out with $T$, and Enter with $N$.

\paragraph{Monetary Policy:} Player $1$ is a central bank, that interacts with a continuum of households (player $2$s), each has negligible mass.
  In every period, the central bank chooses the inflation level while at the same time, households form their expectations about inflation. To simplify matters, I assume that both actual inflation and expected inflation are binary variables. In line with the classic work of Barro (1986), players' stage game payoffs are:
\begin{center}
\begin{tabular}{| c | c | c |}
  \hline
  - & Low Expectation & High Expectation \\
  \hline
  Low Inflation & $1-\theta,x_1$ & $-d(\theta),-y_1$ \\
  \hline
  High Inflation & $1,-y_2$ & $0,x_2$ \\
  \hline
\end{tabular}
\end{center}
where $x_1,x_2,y_1,y_2>0$ are parameters,
$\theta \in \Theta \subset (0,1)$ is the central bank's private information.
This game maps into the simultaneous-move version of my stage game.

To interpret these payoffs, households want to match their expectations with the actual inflation. The central bank's payoff decreases with the actual inflation and increases with the amount of surprised inflation (defined as actual inflation minus expected inflation). As argued in Barro (1986), the central bank can strictly benefit from surprised inflation as it can  increase real economic activities, decrease unemployment rate and increase governmental revenue.
How the central bank trades-off these benefits with the costs of inflation is captured by $\theta$, which depends on the central banker's ideology and tends to be her persistent private information. The assumption  that $\theta <1$ implies that inflation is costly for the central bank if it is fully anticipated by households.
\appendix
\newpage
\section{Proof of Theorem 1: Constructing High-Payoff Equilibria}
\paragraph{Defining Constants:} Recall the definitions of $\gamma^*$, $v^H$, $v^N$, $v^L$, and $v(\gamma)$. There
exists a rational number $\widehat{n}/\widehat{k} \in (\gamma^*,\gamma)$ with $\widehat{n},\widehat{k} \in \mathbb{N}$. Hence, there exists an integer $j \in \mathbb{N}$ such that
\begin{equation*}
  \frac{\widehat{n}}{\widehat{k}} =\frac{\widehat{n}j}{\widehat{k}j} <  \frac{\widehat{n}j}{\widehat{k}j-1} < \gamma.
\end{equation*}
Let $n \equiv \widehat{n}j$ and $k \equiv \widehat{k}j$.
Let $\delta \in (0,1)$ be large enough such that:
\begin{equation}
        \frac{\delta+\delta^2+...+\delta^n}{\delta+\delta^2+...+\delta^k} <
        \widetilde{\gamma}
         <\frac{\delta^{k-n-1}(\delta+\delta^2+...+\delta^n)}{\delta+\delta^2+...+\delta^{k-1}}.
\end{equation}
Later on in the proof, I impose two other requirements on $\delta$, given by (\ref{A.25}). These are compatible with (A.2) since all of these requirements are satisfied when $\delta$ is above some cutoff.
Let
\begin{equation}
    \widetilde{\gamma} \equiv
       \frac{1}{2} \Big(
        \frac{n}{k}+\frac{n}{k-1}
        \Big) \textrm{ and }
    \widehat{\gamma} \equiv       \frac{1}{2} \Big(
        \frac{n}{k}+\gamma^*
        \Big).
\end{equation}
By construction, $\gamma^* <\widehat{\gamma}<\frac{n}{k}< \widetilde{\gamma} <\frac{n}{k-1}<\gamma$.
Let $\pi_j$ be the prior probability of type $\theta_j$.
For every $j \geq 3$, let
$k_j \in \mathbb{N}$ be large enough such that:
\begin{equation}\label{A.5}
   (1-\gamma^* \pi_1) \frac{(\pi_j/k_j)}{\sum_{n=2}^{k} \pi_n +(\pi_j/k_j)}
    \leq 1-\gamma^*.
\end{equation}
Let $K \equiv \sum_{j=3}^m k_j$.
Let $\eta^* \in [\gamma^* \pi_1,\pi_1)$ be large enough such that for every $\eta \in [\eta^*,\pi_1]$, we have:
\begin{equation}\label{A.6}
    \frac{\pi_1-\eta}{\pi_1 (1-\eta)}
    \leq \min_{j \in \{3,...,m\}}
    \Big\{
    \frac{\pi_j/k_j }{\pi_2+...+\pi_j}
    \Big\}
\end{equation}
Let $\lambda \in (0,\frac{1-\sqrt{\gamma^*}}{\gamma^*})$ be small enough such that:
\begin{equation}\label{A.8}
    \big(1-\lambda \gamma^*\big)^{1-\widehat{\gamma}} \big(1+\lambda (1-\gamma^*)\big)^{\widehat{\gamma}} > 1.
\end{equation}

\paragraph{State Variables:} The constructed equilibrium keeps track of three state variables:
\begin{enumerate}
  \item $\eta(h^t)$, which is the probability P2's belief attaches to type $\theta_1$ at $h^t$.
  \item P1's continuation value at $h^t$, $v(h^t) \equiv \{v_{j}(h^t)\}_{j=1}^m$. I show in Appendix A.4.1 that $v(h^t)$ can be written as
  $v(h^t) \equiv p^L(h^t) v^L +p^H (h^t) v^H +p^N (h^t) v^N$.
  \item $\overline{\theta}(h^t)$, which is the highest cost type in the support of P2's belief at $h^t$, and its probability.
\end{enumerate}
The third state variable is implied by the first one when there are \textit{only two types} in the support of buyers' prior belief.
I describe players' actions and the evolution of P1's continuation value at \textit{on-path histories}.
At \textit{off-path histories},
P2 plays $N$ and every type of P1 plays $L$.
I partition the set of on-path histories into three classes:
\begin{itemize}
  \item \textbf{Class 1 Histories:} $h^t$ is such that $p^L(h^t)  \geq 1-\delta$.
  \item \textbf{Class 2 Histories:} $h^t$ is such that $p^L(h^t) \in (0,1-\delta)$.
  \item \textbf{Class 3 Histories:} $h^t$ is such that $p^L(h^t)=0$.
\end{itemize}
Play starts from a Class 1 history $h^0$ and eventually reaches some Class 3 histories. Class 3 histories are absorbing in the sense that if $p^L(h^t)=0$, then $p^L(h^s)=0$ for all $h^s \succeq h^t$.
Active learning about P1's type happens at Class 1 and Class 2 histories, but stops after reaching Class 3 histories.

\subsection{Class 1 Histories}\label{subB.1}
\paragraph{Players' Actions:} At every $h^t$ that satisfies
$p^L(h^t) \geq 1-\delta$:
\begin{itemize}
  \item Player $2$ plays $T$ for sure.
  \item Type $\theta_1$  plays $H$ with probability:
\begin{equation}\label{A.1}
 \frac{\eta(h^t)-\eta(h^t,L)}{\eta(h^t,H)-\eta(h^t,L)} \cdot \frac{\eta(h^t,H)}{\eta(h^t)},
\end{equation}
and other types in the support of P2's belief play $H$ with the same probability, equal to:
\begin{equation}\label{A.2}
    \frac{\eta(h^t)-\eta(h^t,L)}{\eta(h^t,H)-\eta(h^t,L)} \cdot \frac{1-\eta(h^t,H)}{1-\eta(h^t)},
\end{equation}
where the posterior beliefs $\eta(h^t,H)$ and $\eta(h^t,L)$ are functions of $\eta(h^t)$, given by:
 \begin{equation}\label{A.3}
 \eta (h^t,H) = \eta^*+ \min \Big\{1-\eta^*, \big(1+\lambda (1-\gamma^*) \big) \big(\eta (h^t)-\eta^* \big) \Big\},
\end{equation}
\begin{equation}\label{A.4}
 \eta (h^t,L)  = \eta^* + (1-\lambda \gamma^*) (\eta (h^t)-\eta^*),
 \end{equation}
 with $\eta^*$ a constant that satisfies (\ref{A.6}), and $\lambda$ is a constant that satisfies
(\ref{A.8}).
\end{itemize}
One can use (\ref{A.1}), (\ref{A.2}), (\ref{A.3}), and (\ref{A.4}) to write the probability that each type of player $1$ playing $H$ at $h^t$ as a function of $\eta(h^t)$, i.e., P1's action at Class 1 histories only depends on P2's belief about her being type $\theta_1$.

\paragraph{P1's Continuation Value:} For every $h^t$ that satisfies
$p^L(h^t) \geq 1-\delta$:
\begin{enumerate}
  \item If P1 plays $L$ at $h^t$, then his continuation value is:
  \begin{equation}\label{B.9}
v(h^t,L)= \frac{p^N(h^t)}{\delta} v^N +  \frac{p^L(h^t)-(1-\delta)}{\delta} v^L + \frac{p^H(h^t)}{\delta} v^H.
\end{equation}
  \item If $h^t$ is such that $\eta(h^t,H)<1$, P1's continuation value after playing $H$ at $h^t$ is:
\begin{equation}\label{B.10}
v(h^t,H)=\frac{p^N(h^t)}{\delta} v^N  + \frac{p^L(h^t)}{\delta}v^L+ \frac{p^H(h^t)-(1-\delta)}{\delta} v^H.
\end{equation}
If $h^t$ is such that $\eta(h^t,H)=1$, P1's continuation value after playing $H$ at $h^t$ is:
\begin{equation}\label{B.11}
  v(h^{t},H) = \frac{v_{1}(h^{t},H)}{1-\theta_1} v^H +\Big(1-\frac{v_{1}(h^{t},H)}{1-\theta_1}\Big)v^N \in \mathbb{R}^m,
\end{equation}
\begin{equation}\label{B.12}
 \textrm{ with }  v_1(h^{t},H) \equiv \frac{v_1(h^{t})-(1-\delta)(1-\theta_1)}{\delta} \textrm{ and }
 v_1(h^t)\in \mathbb{R} \textrm{ is the first entry of } v(h^t).
\end{equation}
\end{enumerate}
\paragraph{Players' Incentives:} I verify players' incentive constraints at Class 1 histories:
\begin{enumerate}
\item If $p^L(h^t) \geq 1-\delta$ and $\eta(h^t,H)<1$, then according to (\ref{B.9}) and (\ref{B.10}),
all types of P1 are indifferent between playing $H$ and $L$ at $h^t$.
\item If $p^L(h^t) \geq 1-\delta$ and $\eta(h^t,H)=1$, then according to (\ref{B.9}) and  (\ref{B.11}),
type $\theta_1$ is indifferent between $H$ and $L$ at $h^t$, and other types in the support of P2's belief strictly prefer to play $L$ at $h^t$.
\item If P2's beliefs are updated according to
(\ref{A.3}) and (\ref{A.4}), then
$H$ is played at $h^t$ with probability at least $\gamma^*$, i.e., P2 has an incentive to play $T$
 at $h^t$. This is derived in section A.4.
\end{enumerate}
Belief updating formulas (\ref{A.3}) and (\ref{A.4}), together with (\ref{A.8}) lead to the following lemma:
\begin{Lemma}\label{LA.1}
For every $\underline{\eta} \in (\eta^*,1)$,
there exist $T \in \mathbb{N}$ and $\underline{\delta} \in (0,1)$,
s.t. when $\eta (h^r) \geq \underline{\eta}$ and $\delta>\underline{\delta}$,
if $h^t \equiv (y_0,...,y_{t-1}) \succ h^r$ and all histories between $h^r$ and $h^t$ belong to Class 1, then:
\begin{equation}\label{A.9}
    \underbrace{(1-\delta) \sum_{s=r}^{t-1} \delta^{s-r} \mathbf{1}\{y_s=H\}}_{\textrm{weight of $(T,H)$ played from $r$ to $t$}} \leq \underbrace{(1-\delta^T)}_{\textrm{weight of initial $T$ periods}}
    +   \underbrace{ (1-\delta)\sum_{s=r}^{t-1} \delta^{s-r} \mathbf{1}\{y_s=L\}}_{\textrm{weight of $(T,L)$ played from $r$ to $t$}}
    \cdot
    \frac{\widetilde{\gamma}}{1-\widetilde{\gamma}}.
\end{equation}
\end{Lemma}
The proof is in Appendix A.5. For some intuition,
given the belief updating formulas (\ref{A.3}) and (\ref{A.4}),
player $2$'s posterior belief at $h^t$ depends only on her belief at $h^r$ and the number of times $H$ and $L$ have been played from period $r$ to $t$.
Since the choice of $\lambda$ satisfies the first inequality in (\ref{A.8}),
if player $1$ plays $H$ with  (undiscounted) frequency above $\widehat{\gamma}$, then player $2$'s belief at $h^t$ attaches higher probability to type $\theta_1$ compared to her belief at $h^r$. When P2's belief at $h^r$ attaches probability more than $\underline{\eta}$ to type $\theta_1$,
her posterior attaches probability $1$ to type $\theta_1$ before period $r+S$, where:
\begin{equation}\label{A.S}
S \equiv \Big\lceil  \frac{\log \frac{1-\eta^*}{\underline{\eta}-\eta^*}}{\displaystyle \log \Big\{ (1-\lambda \gamma^* )^{1-\widehat{\gamma}} (1+\lambda (1-\gamma^*))^{\widehat{\gamma}} \Big\} } \Big\rceil,
\end{equation}
after which P2's belief about type $\theta_1$ reaches $1$, and
the convex weight of $v^L$ equals $0$
according to (\ref{B.11}).

The requirement that all histories from $h^r$ to $h^t$ belonging to Class 1
not only leads to an upper bound on the \textit{undiscounted} frequency with which
$(T,H)$ being played from $r$ to $t$,
but also imposes constraints on how frontloaded outcome $(T,H)$ can be. For example, after P1 plays $H$ in the first
\begin{equation}\label{A.T}
    T \equiv \Big\lceil
    \frac{\log \frac{1}{\pi_1}}{\log \big(1+\lambda (1-\gamma^*)\big)}
    \Big\rceil
\end{equation}
periods,  P2's belief about type $\theta_1$ reaches $1$, and
the convex weight of $v^L$ equals $0$
according to (\ref{B.11}).
If $\delta$ is large enough, then the constraint on undiscounted frequency and the constraint on frontloadedness of outcome $(T,H)$
lead to an upper bound on the \textit{discounted
frequency} with which outcome
$(T,H)$ occurs from $r$ to $t$, with
$\widehat{\gamma}$ being replaced by a larger $\widetilde{\gamma}$
to provide extra slack caused by the discount factor $\delta$.

I apply Lemma A.1 by setting $h^r=h^0$ and $\underline{\eta}=\eta(h^0)$.
If $h^t$ and all its predecessors belong to Class 1, then:
  \begin{equation*}
  (1-\delta)\sum_{s=0}^{t-1} \delta^s \mathbf{1}\{y_s=(T,L)\} \leq p^L(h^0)=\frac{(1-\theta_1)(1-\gamma)}{1-\gamma \theta_1}.
  \end{equation*}
Lemma A.1 leads to an upper bound on $(1-\delta)\sum_{s=0}^{t-1} \delta^s \mathbf{1}\{y_s=(T,H)\}$,
which implies that if $\delta$ is large enough, then
\begin{equation}\label{A.10}
p^H(h^t) \geq     Y \equiv \frac{1}{2}
  \underbrace{  \Big(
    \gamma
    -(1-\gamma) \frac{\widetilde{\gamma}}{1-\widetilde{\gamma}}
    \Big)}_{>0}
    \frac{1-\theta_1}{1-\gamma \theta_1},
\end{equation}
for every $h^t$ such that $h^t$ and all its predecessors belonging to Class 1.

\subsection{Class 2 Histories}\label{subB.2}
\paragraph{Players' Actions:} If $h^t$ is such that $p^L(h^t) \in (0, 1-\delta)$, then at $h^t$,
\begin{itemize}
\item[1.] Player $2$ plays $T$ for sure.
\item[2.] Types in the support of P2's belief at $h^t$ \textit{except for} type $\overline{\theta}(h^t)$
  play $H$ for sure.
  Type $\overline{\theta}(h^t)$ potentially mixes between $H$ and $L$, with probabilities specified below.
\end{itemize}
 Let
\begin{equation}\label{A.13}
    l(h^t) \equiv \# \Big\{
    h^s \Big| h^s\prec h^t, h^s \textrm{ belongs to Class 2, and } \overline{\theta}(h^s)=\overline{\theta}(h^t)
    \Big\}
\end{equation}
be the number of histories that (1) strictly precede $h^t$, and (2) the highest-cost type in the support of P2's belief is also
$\overline{\theta}(h^t)$. Consider two cases separately, depending on whether $\overline{\theta}(h^t)$ is $\theta_2$ or not.
\begin{enumerate}
  \item If $\overline{\theta}(h^t)=\theta_j$ with $j \geq 3$, then
type $\overline{\theta}(h^t)$ plays $L$ at $h^t$ with probability
\begin{equation}\label{A.14}
    \frac{1}{k_j-l(h^t)},
\end{equation}
in which $k_j$ is the integer defined in (\ref{A.5}).
  \item If $\overline{\theta}(h^t)=\theta_2$, then
type $\overline{\theta}(h^t)$ plays $L$ at $h^t$ with probability
\begin{equation}\label{B.19}
    \min\{1, \frac{1-\gamma^*}{1-\eta(h^t)}\}.
\end{equation}
\end{enumerate}

\paragraph{P1's Continuation Value:} After player $1$ plays $L$ at $h^t$, P1's continuation value is
  \begin{equation}\label{A.11}
v(h^t,L) \equiv   \frac{Q(h^t)}{\delta} v^H + \frac{\delta-Q(h^t)}{\delta} v^N,
  \end{equation}
  where
  \begin{equation}\label{A.12}
    Q(h^t) \equiv p^H(h^t) - \frac{1-\delta-p^L(h^t)}{1-\overline{\theta}(h^t)}
  \end{equation}
After player $1$ plays $H$ at $h^t$, his continuation value depends on whether $\eta(h^t,H)$ equals $1$ or not, with
$\eta(h^t,H)$ computed via Bayes Rule given P2's belief at $h^t$ and type $\overline{\theta}(h^t)$'s mixing probability at $h^t$:
\begin{enumerate}
  \item If $\eta(h^t,H)<1$, then P1's continuation payoff at $(h^t,H)$, denoted by
  $v(h^t,H)$, is given by (\ref{B.10}).
  \item If $\eta(h^t,H)=1$, then
  P1's continuation payoff at $(h^t,H)$, denoted by
  $v(h^t,H)$, is given by (\ref{B.11}).
\end{enumerate}
By construction of player $1$'s equilibrium actions, it is clear that $\eta(h^t,H)=1$ at Class 2 history $h^t$ only when $\overline{\theta}(h^t)=\theta_2$. This is because when $\overline{\theta}(h^t)> \theta_3$, type $\theta_2$ plays $H$ at $h^t$ for sure, and $\eta(h^t,H)<1$.
\paragraph{Players' Incentives:} Lemma A.2 states that players' incentive constraints at Class 2 histories are satisfied.
\begin{Lemma}\label{LA.2}
At every Class 2 history $h^t$,
\begin{enumerate}
  \item P2 has an incentive to play $T$.
  \item If $h^t$ is such that $\eta(h^t,H)<1$, then type $\overline{\theta}(h^t)$ is indifferent between $H$ and $L$ at $h^t$, and types that have strictly lower cost than $\overline{\theta}(h^t)$ strictly prefer to play $H$ at $h^t$.
  \item If $h^t$ is such that $\eta(h^t,H)=1$, then type $\overline{\theta}(h^t)$ strictly prefers to play $L$ at $h^t$, and types that have strictly lower cost than $\overline{\theta}(h^t)$ strictly prefer to play $H$ at $h^t$.
\end{enumerate}
\end{Lemma}

\paragraph{Properties of Class 2 Histories:} I state three properties of Class 2 histories, all of which are shown in section \ref{subB.3}. Lemma \ref{LA.3} establishes a lower bound on P2's posterior belief after observing $H$ at any Class 2 history.
\begin{Lemma}\label{LA.3}
For any Class 2 history $h^t$.
\begin{itemize}
  \item If $\overline{\theta}(h^t)\geq \theta_3$, then $\eta(h^t,H) \geq \eta(h^0)$  and $\eta(h^t,L) =0$.
  \item If $\overline{\theta}(h^t)=\theta_2$, then $\eta(h^t,H) = \min \{1,\frac{\eta(h^t)}{\gamma^*}\}$ and $\eta(h^t,L) =0$.
\end{itemize}
\end{Lemma}
Lemma \ref{LA.4} establishes an upper bound on the number of Class 2 histories along every path of play.
\begin{Lemma}\label{LA.4}
There exist $\underline{\delta} \in (0,1)$ and $M \in \mathbb{N}$,
  such that when $\delta >\underline{\delta}$, the number of Class 2 histories along every path of equilibrium play is at most $M$.
\end{Lemma}
Lemma \ref{LA.5} establishes a uniform lower bound on $p^H(h^t)$ for all Class 1 and Class 2 histories.
\begin{Lemma}\label{LA.5}
There exist $\underline{\delta} \in (0,1)$ and $\underline{Q}>0$, such that when $\delta > \underline{\delta}$, we have $p^H(h^t) \geq \underline{Q}$ for all $h^t$ belonging to Class 1 and Class 2.
\end{Lemma}
Lemma \ref{LA.5} also implies a lower bound on $p^H(h^t)$ if $h^t$ is the \textit{first history} that reaches Class 3, i.e., $h^t$ is such that $p^L(h^t)=0$ and $p^L(h^s)>0$ for all $h^s \prec h^t$.

\subsection{Class 3 Histories}\label{subB.4}
If $h^t$ is such that $p^L(h^t)=0$, then $v(h^t)$ is a convex combination of $v^H$ and $v^N$.
According to Fudenberg and Maskin (1991), stated as
Lemma 3.7.2 in
Mailath and Samuelson (2006, page 99),
when $\delta$ is large enough, there exist
$\{v^t\}_{t=0}^{\infty}$ with $v^t \in \{v^N,v^H\}$ such that
first, $v(h^t)=(1-\delta) \sum_{t=0}^{\infty} \delta^t v^t$,
and second, for every $s \in \mathbb{N}$,
$(1-\delta) \sum_{t=s}^{\infty} \delta^{t-s} v^t$ is $\varepsilon-$close to $v(h^t)$.
Players' continuation play following $h^t$ is given by:
\begin{itemize}
  \item For every $s \in \mathbb{N}$ such that $v^s=v^H$,
  P2 plays $T$ and all types of P1 play $H$ in period $t+s$.
  \item For every $s \in \mathbb{N}$ such that $v^s=v^N$,
    P2 plays $N$ and all types of P1 play $L$ in period $t+s$.
\end{itemize}
\paragraph{Players' Incentives:} P2's incentive at Class 3 histories are trivially satisfied. For P1's incentives,
pick $\varepsilon$ in Lemma 3.7.2 of Mailath and Samuelson (2006) to be small enough.
Lemma \ref{LA.5} implies that
P1's continuation value at every  Class 3 history is no less than $(\underline{Q}/2) v^H +\Big(1-\underline{Q}/2 \Big) v^N$.
When a patient P1 is asked to play $H$, she has a strict incentive to comply
since if she does not comply, then her continuation payoff is $0$; and if she complies, then her continuation payoff is strictly bounded away from $0$.

\subsection{Incentive Constraints \& Promise Keeping Constraints}\label{subB.3}
First, I verify that at every on-path history, P1's continuation payoff is  a convex combination of $v^N$, $v^H$, and $v^L$.
Next, I show that P2 has an incentive to play $T$ at every Class 1 history.
Then, I show Lemmas A.2 to A.5, which together with Lemma A.1 imply the promise keeping condition, that the continuation play delivers every type of player $1$ her promised continuation value at every on-path history.

\subsubsection{P1's Continuation Value}
P1's continuation value in the beginning $v(h^0)$ is a convex combination of $v^N$, $v^H$, and $v^L$.  I show that:
\begin{itemize}
  \item Suppose $h^t$ is an on-path history and $v(h^t)$ is a convex combination of $v^N$, $v^H$, and $v^L$, then for every outcome $y_t \in \{N,H,L\}$ that occurs with positive probability at $h^t$, P1's continuation value after $y_t$, given by $v(h^t,y_t)$, is also a convex combination of $v^N$, $v^H$, and $v^L$.
\end{itemize}
First, consider the case in which $h^t$ belongs to Class 3. Given that $p^L(h^t)=0$, or equivalently, $v(h^t)$ is a convex combination of $v^N$ and $v^H$, the only on-path outcomes at $h^t$ are $N$ and $(T,H)$. As a result, the continuation payoffs $v(h^t,N)$ and $v(h^t,H)$ are both convex combinations of $v^N$ and $v^H$.

Second, consider the case in which $h^t$ belongs to Class 1. There are two possible outcomes at $h^t$: $(T,H)$ and $(T,L)$.
If $h^t$ is such that $\eta(h^t,H) \neq 1$, then
according to (\ref{B.9}) and (\ref{B.10}),
P1's continuation value remains to be a convex combination of $v^N$, $v^H$, and $v^L$.
If $h^t$ is such that $\eta(h^t,H) = 1$,
then
according to (\ref{B.9}) and (\ref{B.11}),
P1's continuation value remains to be a convex combination of $v^N$, $v^H$, and $v^L$.

Third, consider the case in which $h^t$ belongs to Class 2. There are two possible outcomes at $h^t$: $(T,H)$ and $(T,L)$.
If player $1$ plays $L$, then his continuation value is (\ref{A.11}), which is a convex combination of $v^N$ and $v^H$. If he plays $H$, then his continuation value is (\ref{B.10}) if
$\eta(h^t,H) \neq 1$, and is (\ref{B.11}) if $\eta(h^t,H) = 1$. In both cases, $v(h^t,H)$ is a
convex combination of $v^N$, $v^L$, and $v^H$.
\subsubsection{P2's Incentives at Class 1 Histories}
I show that $H$ is played with probability at least $\gamma^*$ at every Class 1 history, which implies that P2 has an incentive to play $T$.
Let $p_H (h^t)$ be the probability that P1 plays $H$ at $h^t$ \textit{according to P2's belief}.
Since P2's belief is a martingale, we have:
\begin{equation*}
    p_H(h^t) \eta (h^t,H) +(1-p_H(h^t)) \eta (h^t,L)=\eta(h^t).
\end{equation*}
The above equality is equivalent to:
\begin{equation*}
    p_H(h^t) \Big(\eta (h^t,H)-\eta(h^t)\Big)+(1-p_H(h^t)) \Big(\eta (h^t,L)-\eta(h^t)\Big)=0
    \end{equation*}
    \begin{equation*}
    \Leftrightarrow \quad
     p_H(h^t) \Big(\eta (h^t,H)-\eta(h^t)\Big)=(1-p_H(h^t)) \Big(\eta(h^t)-\eta (h^t,L)\Big).
\end{equation*}
As long as $\eta(h^t,L) \neq \eta(h^t)$ and $p_H(h^t) \neq 0$, i.e., nontrivial learning happens at $h^t$, and $H$ is played at $h^t$ with positive probability, we have:
\begin{equation}
    \frac{\eta (h^t,H)-\eta(h^t)}{\eta(h^t)-\eta (h^t,L)}=\frac{1-p_H(h^t)}{p_H(h^t)}.
\end{equation}
If P2 plays $T$ with positive probability at $h^t$, then $p_H(h^t) \geq \gamma^*$. This implies that:
\begin{equation}
    \frac{\eta (h^t,H)-\eta(h^t)}{\eta(h^t)-\eta (h^t,L)}=\frac{1-p_H(h^t)}{p_H(h^t)} \leq \frac{1-\gamma^*}{\gamma^*}.
\end{equation}
The belief updating formulas in (\ref{A.3}) and (\ref{A.4}) satisfy (A.25), and therefore, P2 has an incentive to play $T$.
\subsubsection{Proof of Lemma A.2}
Let $\pi(h^t) \in \Delta (\Theta)$ be P2's belief at $h^t$. For every $\theta \in \Theta$, let $\pi(h^t)[\theta]$ be the probability it attaches to type $\theta$.
A useful observation from the constructed strategies is: for every Class 2 history $h^t$, and every $\theta_i<\theta_j$,
\begin{enumerate}
  \item if $\theta_j$ belongs to the support of P2's belief at $h^t$, then $\theta_i$ also belongs to the support of that belief.
  \item if $\overline{\theta}(h^t)=\theta_j$, then
  \begin{equation*}
    \frac{\pi(h^t)[\theta_j]}{\pi(h^t)[\theta_i]}=\frac{\pi_j}{\pi_i} \cdot \frac{k_j-l(h^t)}{k_j}.
  \end{equation*}
\end{enumerate}
I start from verifying P2's incentives using the observation that at every history $h^t$ belonging to Class 1 or Class 2,
\begin{equation*}
\eta(h^t) \underbrace{\geq}_{\textrm{by induction on $t$}} \eta^* \underbrace{\geq}_{\textrm{according to (\ref{A.6})}} \gamma^* \eta(h^0).
\end{equation*}
Suppose $\overline{\theta}(h^t)=\theta_2$, then only types $\theta_1$ and $\theta_2$ can occur with positive probability at $h^t$. Since
type $\theta_2$ plays $L$ at $h^t$ with probability
    $\min \{1,\frac{1-\gamma^*}{1-\eta(h^t)}\}$,
type $\theta_1$ plays $H$ for sure, and the probability of type $\theta_1$ is $\eta(h^t)$,
player $2$ believes that $L$ is played at $h^t$ with probability at most $1-\gamma^*$. This implies her incentive to play $T$ at $h^t$.

Next, I examine the case in which $\overline{\theta}(h^t)=\theta_j$ with $j \geq 3$.
By definition, types with cost higher than $\theta_j$
occur with probability $0$, and type $\theta_1$ occurs with probability at least $\gamma^* \pi_1$.
According to player $1$'s actions at Class 2 histories specified in section \ref{subB.2}, and using statement 2 in Claim 1,
the probability with which $L$ is played at $h^t$ is at most:
\begin{equation}\label{A.15}
  (1-\gamma^* \pi_1) \frac{(\pi_j/k_j)}{\pi_2+...+\pi_{j-1}+((k_j-l(h^t))\pi_j/k_j)}
\leq
    (1-\gamma^* \pi_1) \frac{(\pi_j/k_j)}{\pi_2+...+\pi_{j-1}+(\pi_j/k_j)}.
\end{equation}
The RHS is no more than $1-\gamma^*$ according to the definition of $k_j$ in (\ref{A.5}).
To verify P1's incentives, I consider two subcases:
\begin{enumerate}
  \item If $h^t$ is such that $\eta(h^t,H)<1$, then (\ref{B.10}), (\ref{A.11}) and (\ref{A.12}) imply that type $\overline{\theta}(h^t)$ is indifferent between $H$ and $L$ at $h^t$, and types that are strictly lower than $\overline{\theta}(h^t)$ strictly prefer $H$ to $L$.
  \item If $h^t$ is such that $\eta(h^t,H)=1$, then given that all types except for type $\overline{\theta}(h^t)$ play $H$ with probability $1$ at $h^t$, then we know that $\overline{\theta}(h^t)=\theta_2$. According to (\ref{B.10}), (\ref{A.11}) and (\ref{A.12}), type $\theta_2$ strictly prefers $L$ at $h^t$, and type $\theta_1$ strictly prefers $H$ at $h^t$.
\end{enumerate}

\subsubsection{Proof of Lemma A.3}
\paragraph{Case 1:} Consider the case in which  $\overline{\theta}(h^t) \geq \theta_3$. First, suppose $\eta(h^t) \geq \eta(h^0)$, then
the conclusion of
Lemma A.3 follows since
$\eta(h^t,H)>\eta(h^t) \geq \eta(h^0)$.
Second, suppose $\eta(h^t)< \eta(h^0)$, then given the value of $l(h^t)$
and the highest-cost type at $h^t$ being $\theta_j$,
the posterior probability of type $\theta_1$ is bounded from below by:
\begin{eqnarray*}
 & &   \frac{\eta(h^t)}{\displaystyle \eta(h^t) +(1-\eta(h^t)) \frac{\pi_2+...+\pi_{j-1} + \frac{k_j-l(h^t)-1}{k_j} \pi_j}{\pi_2+...+\pi_{j-1} + \frac{k_j-l(h^t)}{k_j} \pi_j}}
\geq
\frac{\eta(h^t)}{\displaystyle \eta(h^t) +(1-\eta(h^t)) \frac{\pi_2+...+\pi_{j-1} + \frac{k_j-1}{k_j} \pi_j}{\pi_2+...+\pi_{j-1} +\pi_j}}
\end{eqnarray*}
Let
\begin{equation*}
X \equiv 1- \frac{\pi_2+...+\pi_{j-1} + \frac{k_j-1}{k_j} \pi_j}{\pi_2+...+\pi_{j-1} +\pi_j}=
\frac{\pi_j}{k_j (\pi_2+...+\pi_{j-1} +\pi_j)}.
\end{equation*}
The lower bound on posterior belief $\frac{\eta(h^t)}{\eta(h^t)+(1-\eta(h^t)) (1-X)}$ is greater than $\pi_1$ if and only if:
\begin{equation*}
    X \geq 1-\frac{(1-\pi_1)\eta(h^t)}{\pi_1 (1-\eta(h^t))}
    =\frac{\pi_1-\eta(h^t)}{\pi_1 (1-\eta(h^t))}.
\end{equation*}
Given that $\eta(h^t)\geq \eta^*$ at every history $h^t$ that belongs to Class 2,  the above inequality is implied by (\ref{A.6}).

\paragraph{Case 2:}
Consider the case in which $\overline{\theta}(h^t)=\theta_2$. If $\eta(h^t) \geq \gamma^*$, then type $\theta_2$
plays $L$ with probability $\min \{1,\frac{1-\gamma^*}{1-\eta(h^t)}\}=1$, which implies that $\eta(h^t,H)=1$.
If $\eta(h^t) < \gamma^*$, then type $\theta_2$
plays $L$ with probability $\min \{1,\frac{1-\gamma^*}{1-\eta(h^t)}\}=\frac{1-\gamma^*}{1-\eta(h^t)}$, which implies that $\eta(h^t,H)=\eta(h^t)/\gamma^* \geq \gamma^* \eta(h^0)/\gamma^*=\eta(h^0)$.

\subsubsection{Proof of Lemma A.4}
\paragraph{Step 1:} If $h^t$ belongs to Class 2 and $\overline{\theta}(h^t) =\theta_j \geq \theta_3$, then according to (\ref{A.14}),
type $\overline{\theta}(h^t)$ plays $L$ with probability $1$
when $l(h^t)=k_j-1$, after which play reaches a Class 3 history. Therefore, along every path of play, there are at most $k_j$ Class 2 histories satisfying $\overline{\theta}(h^t)=\theta_j$, and there are at most $K \equiv k_3+...+k_m$ Class 2 histories that has $\overline{\theta}(h^t) \geq \theta_3$.

\paragraph{Step 2:} Let $h^t$ be a Class 2 history with $\overline{\theta}(h^t)=\theta_2$.
Let $N \equiv \lceil \frac{1}{1-\gamma}\rceil$, and recall $T$ in Lemma A.1. In addition to the requirements on $\delta$ mentioned earlier, I also require $\delta$ to satisfy:
\begin{equation}\label{A.25}
    \delta^{T+1} (1+\delta+...+\delta^N) > N \textrm{ and } 2 \delta^{T+N+2} > 1.
\end{equation}
These are compatible given that all of them require $\delta$ to be sufficiently large.

First, I show that after P1 plays $H$ at $h^t$, it takes at most $T+N$ periods for play to reach a history that belongs to \textit{either Class 2 or Class 3}. According to the continuation value at $(h^t,H)$, given by (\ref{B.10}), we have:
\begin{equation}\label{A.26}
p^L(h^t,H)=\frac{p^L(h^t)}{\delta} < \frac{1-\delta}{\delta}.
\end{equation}
The last inequality comes from $h^t$ belonging to Class 2, so that $p^L(h^t)<1-\delta$ by definition.
According to Lemma A.1, for every Class 1 history $h^s$ such that $h^s \succ (h^t,H)$ and all histories between $(h^t,H)$ and $h^s$ belong to Class 1, \begin{equation}\label{A.27}
    (1-\delta) \sum_{r={t+1}}^{s} \delta^{r-(t+1)} \mathbf{1}\{y_r=(T,H)\}
    \leq (1-\delta^T) +     (1-\delta) \sum_{r={t+1}}^{s} \delta^{r-(t+1)} \mathbf{1}\{y_r=(T,L)\} \frac{\widetilde{\gamma}}{1-\widetilde{\gamma}}
\end{equation}
Moreover, (\ref{A.26}) and the requirement that all histories between $(h^t,H)$ and $h^s$ belong to Class 1 imply that
\begin{equation}\label{A.28}
    (1-\delta) \sum_{r={t+1}}^{s} \delta^{r-(t+1)} \mathbf{1}\{y_r=(T,L)\}
    < \frac{1-\delta}{\delta}.
\end{equation}
Given that only outcomes $(T,L)$ and $(T,H)$ occur at Class 1 and Class 2 histories:
\begin{equation*}
    1-\delta^{s-(t+1)}
    =   (1-\delta) \sum_{r={t+1}}^{s} \delta^{r-(t+1)} \mathbf{1}\{y_r=(T,L)\}
    + (1-\delta) \sum_{r={t+1}}^{s} \delta^{r-(t+1)} \mathbf{1}\{y_r=(T,H)\}
\end{equation*}
\begin{equation}\label{A.29}
    \leq (1-\delta^T)+ \frac{1-\delta}{\delta}+ \frac{1-\delta}{\delta} \frac{\widetilde{\gamma}}{1-\widetilde{\gamma}}
    \leq (1-\delta^T) +\frac{1-\delta}{\delta} \frac{1}{1-\widetilde{\gamma}} \leq (1-\delta^T) +\frac{1-\delta}{\delta} \frac{1}{1-\gamma}
\end{equation}
To show that $s-(t+1) \leq T+N$, suppose toward a contradiction that $s-(t+1) \geq T+N+1$, then
\begin{equation*}
 (1-\delta^T) +\frac{1-\delta}{\delta}N \geq (1-\delta^T) +\frac{1-\delta}{\delta} \frac{1}{1-\gamma}\geq  1-\delta^{s-(t+1)} \geq 1-\delta^{T+N+1},
\end{equation*}
which yields:
 \begin{equation*}
\frac{1-\delta}{\delta}N \geq  \delta^T(1-  \delta^{N+1}).
 \end{equation*}
Dividing both sides by $\frac{1-\delta}{\delta}$, we have:
\begin{equation*}
 N \geq    \delta^{T+1} (1+\delta+...+\delta^N),
\end{equation*}
which contradicts the first inequality of (\ref{A.25}). The above contradiction implies that  $s-(t+1) \leq T+N$.

Second, I focus on history $h^s$ that has the following two features:
\begin{itemize}
  \item[1.] $h^s$ belongs to Class 2,
  \item[2.] $h^s \succeq (h^t,H)$ and
  all histories between $(h^t,H)$ and $h^s$, excluding $h^s$, belong to Class 1.
\end{itemize}
I show that there exists at most one period from $(h^t,H)$ to $h^s$ such that the stage-game outcome is $(T,L)$. Suppose toward a contradiction that there exist two or more such periods, then
\begin{equation*}
    (1-\delta) \sum_{r={t+1}}^{s} \delta^{r-(t+1)} \mathbf{1}\{y_r=(T,L)\} \geq 2(1-\delta)\delta^{T+N+1}.
\end{equation*}
The last inequality comes from the previous conclusion that $s -(t+1) \leq T+N$. This is because $h^s$ belongs to Class 2 and $h^{s-1}$ belongs to Class 1, and therefore, $(s-1)-(t+1) \leq T+N$, or equivalently, $s-(t+1) \leq T+N+1$.
According to (\ref{A.28}),
\begin{equation}\label{A.30}
2(1-\delta)\delta^{T+N+1}  <  (1-\delta) \sum_{r={t+1}}^{s} \delta^{r-(t+1)} \mathbf{1}\{y_r=(T,L)\} < \frac{1-\delta}{\delta}.
\end{equation}
The above inequality contradicts the second inequality of (\ref{A.25}) that $2 \delta^{T+N+2} > 1$.

Let $h^t$ be the first time play reaches a history that belongs to Class 2 with $\overline{\theta}(h^t)=\theta_2$. According to Lemma A.3, $\eta(h^t,H) \geq \frac{\eta^*}{\gamma^*} \geq \eta(h^0)=\pi_1$.
Let $h^s$ be the next history that belongs to Class 2 with $h^s \succ (h^t,H)$. Since we have shown that $(T,L)$ occurs at most once between $(h^t,H)$ and $h^s$, we know that
\begin{equation*}
    \eta(h^s, H) =\min\{1, \frac{\eta(h^s)}{\gamma^*} \}
    \geq \min \{1, \frac{\eta(h^t,H)}{\gamma^*} (1-\lambda \gamma^*)\}
\end{equation*}
Therefore, conditional on $(h^s, H)$ is not a Class 3 history, player $2$'s belief at $(h^s, H)$ attaches
probability at least:
\begin{equation}\label{A.31}
 \eta(h^s, H) \geq   \eta(h^t,H) \frac{1-\lambda \gamma^*}{\gamma^*} \geq  \eta(h^t,H) \sqrt{\frac{1}{\gamma^*}}
\end{equation}
to type $\theta_1$,
where the last inequality comes from $\lambda \in (0,\frac{1-\sqrt{\gamma^*}}{\gamma^*})$.
Let
\begin{equation*}
    \widehat{M} \equiv \frac{\log (1/\pi_1)}{\log \sqrt{\frac{1}{\gamma^*}}}+1.
\end{equation*}
Since $\eta(h^t,H) \geq \pi_1$ for the first Class $2$ history $h^t$ satisfying $\overline{\theta}(h^t)=\theta_2$, there can be at most $\widehat{M}$ Class 2 histories
with $\theta_2$ being the highest-cost type along every path of play. This is because otherwise,
P2's posterior belief attaches probability greater than
\begin{equation*}
    \pi_1 \Big(\frac{1}{\sqrt{\gamma^*}}\Big)^{\widehat{M}} >1
\end{equation*}
at the $\widehat{M}+1$th such history, which leads to a contradiction.
Summarizing the conclusions of the two parts, there exist at most $M \equiv K+\widehat{M}$ Class 2 histories along every path of equilibrium play.

\subsubsection{Proof of Lemma A.5}
To start with,
consider Class 2 history $h^t$ such that no predecessor of $h^t$ belongs to Class 2, in another word, all predecessors of $h^t$ belong to Class 1. According to (\ref{A.10}), $p^H(h^{t-1}) \geq Y$, which implies that $p^H(h^{t}) \geq Y-(1-\delta)$. As a result
\begin{equation*}
    Q(h^t)=p^H(h^t) - \frac{1-\delta-p^L(h^t)}{1-\overline{\theta}(h^t)} \geq Y-(1-\delta)\big(
    1+ \frac{1}{1-\theta_m}
    \big) >0.
\end{equation*}
If play remains at Class 1 or Class 2 history after $h^t$, then player $1$ must be playing $H$ at $h^t$, after which
\begin{equation*}
    p^H(h^t,H) \geq p^H(h^t)-(1-\delta) \geq Y-2(1-\delta) \textrm{ and } p^L(h^t,H) \leq \frac{1-\delta}{\delta}.
\end{equation*}
According to Lemma A.3, $\eta(h^t,H) \geq \eta(h^0)=\pi_1$. One can then apply Lemma A.1 again, which implies that at every Class $1$ history $h^s$ such that only one predecessor of $h^s$ belongs to Class 2, we have:
\begin{equation*}
    p^H(h^s) \geq Z \equiv Y-2(1-\delta)-\frac{1-\delta}{\delta}\frac{\widetilde{\gamma}}{1-\widetilde{\gamma}}-(1-\delta^T),
\end{equation*}
with $T $ and $\widetilde{\gamma}$ being the same as in the previous step. When $\delta$ is large enough, $Z \geq Y/2$. One can then show that for every Class 2 history $h^s$ such that there is only one strict predecessor history belongs to Class 2,
\begin{equation*}
        Q(h^s)=p^H(h^s) - \frac{1-\delta-p^L(h^s)}{1-\overline{\theta}(h^s)} \geq Z-(1-\delta)\big(
    1+ \frac{1}{1-\theta_m}
    \big) >0.
\end{equation*}
Iteratively apply this process. Since
\begin{enumerate}
  \item the number of Class 2 histories along every path of play is bounded from above by $M$ (Lemma A.4),
  \item for every Class 2 history $h^t$, $p^L(h^t,H) =\frac{1-\delta}{\delta}$ and
 $\eta(h^t,H) \geq \eta(h^0)$,
\end{enumerate}
there exist $\underline{\delta} \in (0,1)$ and $\underline{Q}>0$ such that when $\delta >\underline{\delta}$, $p^H(h^t) \geq \underline{Q}$ for every Class 1 or Class 2 history $h^t$.

\subsection{Proof of Lemma A.1}
For every $h^t$, let $\Delta (h^t) \equiv \eta(h^t)-\eta^*$.
For every $t \in \mathbb{N}$, let $N_{L,t}$ and $N_{H,t}$ be the number of periods in which $L$ and $H$ are played from period $0$ to $t-1$, respectively. The proof is done by induction on $N_{L,t}$.

When $N_{L,t} \leq 2(k-n)$, then the conclusion holds as $N_{H,t} \geq 2n+X$. Moreover, $\Delta (h^T)$ reaches $1-\eta^*$ before period $T$, after which play reaches a Class 3 history.

Suppose the conclusion holds for when $N_{L,t} \leq N$ with $N \geq 2(k-n)$, and suppose toward a contradiction that there exists $h^T$ with $T \geq k+X$ and $N_{L,T} =N+1$, such that every $h^t \preceq h^T$ belongs to Class 1,
but
\begin{equation}\label{A.41}
    (1-\delta) \sum_{t=0}^{T-1} \delta^t \mathbf{1}\{y_t=H\}-
    (1-\delta^X)
    >   (1-\delta)\sum_{t=0}^{T-1} \delta^t \mathbf{1}\{y_t=L\}
    \cdot
    \frac{\widetilde{\gamma}}{1-\widetilde{\gamma}},
\end{equation}
I obtain a contradiction in three steps.

\paragraph{Step 1:} I show that for every $s<T$,
\begin{equation}\label{A.42}
     (1-\delta) \sum_{t=s}^{T-1} \delta^t \mathbf{1}\{y_t=H\}
    \geq   (1-\delta)\sum_{t=s}^{T-1} \delta^t \mathbf{1}\{y_t=L\}
    \frac{\widetilde{\gamma}}{1-\widetilde{\gamma}}.
\end{equation}
Suppose toward a contradiction that (\ref{A.42}) fails.
Then together with (\ref{A.41}), we have:
\begin{equation}\label{A.43}
    (1-\delta) \sum_{t=0}^{s-1} \delta^t \mathbf{1}\{y_t=H\}-
    (1-\delta^X)
    >   (1-\delta)\sum_{t=0}^{s-1} \delta^t \mathbf{1}\{y_t=L\}
    \frac{\widetilde{\gamma}}{1-\widetilde{\gamma}}
\end{equation}
and
\begin{equation}\label{A.44}
     (1-\delta)\sum_{t=s}^{T-1} \delta^t \mathbf{1}\{y_t=L\}>0.
\end{equation}
According to (\ref{A.44}), $N_{L,s}< N_{L,T}$. Since $N_{L,T}=N+1$, we have $N_{L,s} \leq N$. Applying the induction hypothesis and (\ref{A.43}), we know that play reaches a Class 3 history before $h^s$, leading to a contradiction.

\paragraph{Step 2:} I show that for every $k$ consecutive periods
\begin{equation*}
    \{y_r,y_{r+1},...,y_{r+k-1}\} \subset h^T,
\end{equation*}
the number of outcome $(T,H)$ in this sequence is at least $n+1$.
According to (\ref{A.42}) shown in the previous step, outcome $(T,H)$ occurs at least $n+1$ times in the last $k$ periods, namely, in the set $\{y_{T-k+1},...,y_T\}$.

Suppose toward a contradiction that there exists $k$ consecutive periods in which outcome $(T,H)$ occurs no more than $n$ times, then the conclusion above that outcome $(T,H)$ occurs at least $n+1$ times in the last $k$ periods implies that there exists
$k$ consecutive periods
$\{y_r,...,y_{r+k-1}\}$
in which $(T,H)$ occurs exactly $n$ times and $(T,L)$ occurs exactly $k-n$ times. According to (A.2), we have
\begin{equation}\label{A.45}
    (1-\delta) \sum_{t=r}^{r+k-1} \delta^t \mathbf{1}\{y_t=H\}
    <  (1-\delta) \sum_{t=r}^{r+k-1} \delta^t \mathbf{1}\{y_t=L\} \frac{\widetilde{\gamma}}{1-\widetilde{\gamma}},
\end{equation}
but then
\begin{equation}\label{A.46}
    \Delta (h^{r+k}) > \Delta (h^{r+1}).
\end{equation}
Next, let us consider the following new sequence with length $T-k$:
\begin{equation*}
 \widetilde{h}^{T-k} \equiv  \{\widetilde{y}_0,\widetilde{y}_1,...,\widetilde{y}_{T-k-1}\}  \equiv \{y_0,y_1,...,y_{r-1},y_{r+k},...,y_{T-1}\}
\end{equation*}
which is obtained by removing $\{y_r,...,y_{r+k-1}\}$ from the original sequence and front-loading the subsequent play $\{y_{r+k},...,y_{T-1}\}$. The number of $(T,L)$ in this new sequence is at most $N+1-(n-k)$, which is no more than $N$. According to the conclusion in Step 1:
\begin{equation}\label{A.47}
    (1-\delta) \sum_{t=r+k}^{T-1} \delta^t \mathbf{1}\{y_t=H\}
    >   (1-\delta)\sum_{t=r+k}^{T-1} \delta^t \mathbf{1}\{y_t=L\}
    \frac{\widetilde{\gamma}}{1-\widetilde{\gamma}}.
\end{equation}
This together with (\ref{A.45}) and (\ref{A.41}) imply that
\begin{equation*}
    (1-\delta) \sum_{t=0}^{T-k-1} \delta^t \mathbf{1}\{\widetilde{y}_t=H\}-
    (1-\delta^X)
    >   (1-\delta)\sum_{t=0}^{T-k-1} \delta^t \mathbf{1}\{\widetilde{y}_t=L\}
    \frac{\widetilde{\gamma}}{1-\widetilde{\gamma}}.
\end{equation*}
According to the induction hypothesis, play will reach a Class 3 history before period $T-k$ if player $1$ plays according to $\{\widetilde{y}_0,\widetilde{y}_1,...,\widetilde{y}_{T-k-1}\}$.
\begin{enumerate}
  \item Suppose $\widetilde{h}^{T-k}$ reaches a Class 3 history before period $r$, then play will also reach a Class 3 history before period $r$ according to the original sequence.
  \item Suppose $\widetilde{h}^{T-k}$ reaches a Class 3 history in period $s$, with $s>t$,
then according to (\ref{A.46}), we have $\Delta (\widetilde{h}^s) \leq \Delta (h^{s+k})$. This implies that
play will reach a Class 3 history in period $s+k$ according to the original sequence.
\end{enumerate}
This contradicts the hypothesis that play has never reached a Class 3 history before $h^T$.

\paragraph{Step 3:} For every history $h^T \equiv \{y_0,y_1,...,y_{T-1}\} \in \{H,L\}^T$ and $t \in \{1,...,T-1\}$, define the operator $\Omega_t: \{H,L\}^T  \rightarrow \{H,L\}^T$ as:
\begin{equation}\label{A.48}
    \Omega_t (h^T) =(y_0,...,y_{t-2},y_t,y_{t-1},y_{t+1},...,y_{T-1}),
\end{equation}
in another word, swapping the order between $y_{t-1}$ and $y_t$. Recall the belief updating formula in Class 1 histories and let
\begin{equation}\label{A.49}
    \mathcal{H}^{T,*} \equiv \Big\{h^T \Big|
    \Delta (h^t) < 1-\eta^* \textrm{ for all } h^t \prec h^T
    \Big\}.
\end{equation}
If $h^T \in \mathcal{H}^{T,*}$, then $\Omega_t (h^T) \in \mathcal{H}^{T,*}$ unless:
\begin{itemize}
  \item $y_{t-1}=L$, $y_t=H$.
  \item \textit{and}, $\Big(1+\lambda (1-\gamma^*)\Big)
  \Delta (h^{t-1}) \geq 1-\eta^*$.
\end{itemize}
Next, I show that the above situation cannot occur besides in the last $k$ periods. Suppose toward a contradiction that there exists $t \leq T-k$ such that $h^T \in \mathcal{H}^{T,*}$ but $\Omega_t (h^T) \notin \mathcal{H}^{T,*}$. Based on the conclusion in Step 2, outcome $(T,H)$ occurs at least $n+1$ times in the sequence $\{y_t,...,y_{t+k-1}\}$.
Consider another sequence $\{y_{t-1},...,y_{t+k-1}\}$, in which outcome $(T,H)$ occurs at least $n+1$ times and outcome $(T,L)$ occurs at most $k-n$ times. This implies that:
\begin{eqnarray}\label{A.50}
    \Delta (h^{t+k}) &\geq&  \Delta (h^{t-1}) \Big(1+\lambda (1-\gamma^*)\Big)^{n+1}
    \Big(1-\lambda \gamma^*\Big)^{k-n}
     {}
\nonumber\\
&=& \Delta (h^{t-1}) \underbrace{\Big(1+\lambda (1-\gamma^*)\Big)^{n}
    \Big(1-\lambda \gamma^*\Big)^{k-n}}_{\geq 1} \Big(1+\lambda (1-\gamma^*)\Big)
         {}
\nonumber\\
& \geq & \Delta (h^{t-1}) \Big(1+\lambda (1-\gamma^*)\Big)
 {}
\nonumber\\
&\geq&
1-\eta^*,
\end{eqnarray}
where second inequality follows from $n/k > \widehat{\gamma}$, and the 3rd inequality follows from the hypothesis that $\Omega_t (h^T) \notin \mathcal{H}^{T,*}$. Inequality (\ref{A.50}) implies that play reaches the high phase before period $t+k \leq T$, contradicting the hypothesis that $h^T \in \mathcal{H}^{T,*}$.

To summarize, for every $t \leq T-k$, if $h^T \in \mathcal{H}^{T,*}$, then $\Omega_t (h^T) \in \mathcal{H}^{T,*}$. For every $t > T-k$, if $h^T \in \mathcal{H}^{T,*}$, then $\Omega_t (h^T) \in \mathcal{H}^{T,*}$ unless $y_{t-1}=L$ and $y_t=H$.
Therefore, one can freely front-load outcome $(T,H)$ from period $0$ to $T-k-1$ and obtain the following revised sequence:
\begin{equation}\label{A.51}
    \{H,H,...H,L,...,L,y_{T-k},...,y_{T-1}\},
\end{equation}
which meets the following two requirements. First, the revised sequence (\ref{A.51}) still belongs to set $\mathcal{H}^{T,*}$. Second, the sequence in (\ref{A.51}) satisfies (\ref{A.41}).

According to the conclusion in Step 2, the number of outcome $(T,L)$ from period $0$ to $T-k-1$ cannot exceed $k-n-1$, and
the number of outcome $(T,L)$ from period $T-k$ to $T-1$ cannot exceed $k-n-1$.
This is because otherwise, there exists a sequence of length $k$ that has at most $n$ periods of outcome $(T,H)$, contradicting the two conditions that the revised sequence in (\ref{A.51}) must satisfy. Therefore, the total number of outcome $(T,L)$ in this sequence is at most $2(k-n-1)$. This contradicts the induction hypothesis that the number of outcome $(T,L)$ exceeds $2(k-n)$.

\section{Proof of Theorem 1:  Necessity of Constraints}
In section B.1, I establish the necessity of constraint (\ref{1.2}). In section B.2, I establish the necessity of constraint (\ref{1.3}). In section B.3, I show Lemma \ref{L4.1}, namely, $v_j^*$ is the value of the constrained optimization problem.
\subsection{Necessity of Constraint (\ref{1.2})}
For every strategy profile $\sigma$, let $\mathcal{H}^{\sigma}$ be the set of histories that occur with positive probability under $\sigma$. For every $h^t \in \mathcal{H}^{\sigma}$, let $\Theta^{\sigma}(h^t) \subset \Theta$ be the support of player $2$'s belief at $h^t$.
The necessity of (\ref{1.2}) is implied by:
\begin{Proposition}
For every prior belief $\pi$, including those that \textit{do not have full support},
if type $\theta_i$ is
the lowest-cost type in the support of this prior belief, then
her equilibrium payoff is no more than $1-\theta_i$ in all BNEs.
\end{Proposition}
\begin{proof} Rank player $1$'s actions according to $H \succ L$. Given strategy profile $\sigma$ and history $h^t \in \mathcal{H}^{\sigma}$, let
\begin{equation}\label{B.1}
\overline{a}_{1}^{\sigma}(h^{t}) \equiv \max  \Big\{\bigcup_{\theta \in \Theta^{\sigma}(h^t)} \textrm{supp} \Big(\sigma_{\theta}(h^t) \Big) \Big\}
\end{equation}
be the highest action played by player $1$ with positive probability at $h^t$. By definition,
for every BNE $\sigma$ and $h^t \in \mathcal{H}^{\sigma}$, if $\sigma_2(h^t)$ assigns positive probability to $T$, then
$\overline{a}_1^{\sigma}(h^t)=H$.
The rest of
my proof is done by induction on
the number of types in the support of player $2$'s prior belief.
\begin{itemize}
  \item[1.] I establish the conclusion when $|\Theta^{\sigma}(h^0)|=1$.
  \item[2.] Suppose the conclusion holds when $|\Theta^{\sigma}(h^0)| \leq n$, it also holds when $|\Theta^{\sigma}(h^0)|=n+1$.
\end{itemize}

\paragraph{Step 1:} I show that when $|\Theta^{\sigma}(h^0)|=1$, the only type in the support of player $2$'s prior belief, denoted by $\theta_i$, receives payoff no more than $1-\theta_i$. This also implies that for every equilibrium $\sigma$ and for every $h^t \in \mathcal{H}^{\sigma}$, if $\Theta^{\sigma}(h^t)=\{\theta_i\}$ for some $\theta_i \in \Theta$, then type $\theta_i$'s continuation payoff at $h^t$ cannot exceed $1-\theta_i$.

This is because $\Theta^{\sigma}(h^0)=\{\theta_i\}$ implies that $\Theta^{\sigma}(h^t)=\{\theta_i\}$ for every $h^t \in \mathcal{H}^{\sigma}$. Therefore, $\overline{a}_1^{\sigma}(h^t)$ is played by type $\theta_i$ with positive probability at every $h^t \in \mathcal{H}^{\sigma}$. Given type $\theta_i$'s equilibrium strategy $\sigma_{\theta_i}$, the following strategy $\widetilde{\sigma}_{\theta_i}: \mathcal{H} \rightarrow \Delta (A_1)$, defined as:
\begin{equation}\label{B.2}
    \widetilde{\sigma}_{\theta_i}(h^t) \equiv \left\{ \begin{array}{ll}
\overline{a}_1^{\sigma}(h^t) & \textrm{ if } h^t \in \mathcal{H}^{\sigma} \\
\sigma_{\theta_i}(h^t) & \textrm{ otherwise}.
\end{array} \right.
\end{equation}
also best replies against player $2$'s equilibrium strategy $\sigma_2$, from which type $\theta_i$ receives his equilibrium payoff.
If type $\theta_i$ plays according to $\widetilde{\sigma}_{\theta_i}$ against $\sigma_2$,
then according to Step 1, the outcome at every history in $\mathcal{H}^{\sigma}$ is either $(T,H)$ or $N$. Therefore,
type $\theta_i$'s stage-game payoff at every history in $\mathcal{H}^{\sigma}$ cannot exceed $1-\theta_i$, so his discounted average payoff cannot exceed $1-\theta_i$.

\paragraph{Step 2:} I show that if the conclusion holds when $|\Theta^{\sigma}(h^0)| \leq n$, then it also holds when $|\Theta^{\sigma}(h^0)|=n+1$.
I define $\overline{\mathcal{H}}_t^{\sigma}$ for every $t \in \mathbb{N}$ recursively.
Let $\overline{\mathcal{H}}_0^{\sigma} \equiv \{h^0\}$. Given the definition of $\overline{\mathcal{H}}^{\sigma}_{t}$, let
\begin{equation*}
\overline{\mathcal{H}}^{\sigma}_{t+1} \equiv
\Big\{h^{t+1} \in \mathcal{H}^{\sigma} \Big|\exists h^t \in \overline{\mathcal{H}}^{\sigma}_t \textrm{ s.t. } h^{t+1} \succ h^t \textrm{ and either } h^{t+1}=(h^t,N) \textrm{ or } h^{t+1}=\big(h^t,(T,\overline{a}_1^{\sigma}(h^t))\big) \Big\}.
\end{equation*}
Intuitively, $\overline{\mathcal{H}}^{\sigma}_{t+1}$ is the set of period $t+1$ on-path histories such that player $1$ has played his \textit{highest on-path action} from period $0$ to $t$.
Let $\overline{\mathcal{H}}^{\sigma} \equiv \cup_{t=0}^{\infty}\overline{\mathcal{H}}^{\sigma}_{t}$.

Recall that $\theta_i$ is the notation for the lowest-cost type in the support of player $2$'s prior belief.
Given type $\theta_i$'s equilibrium strategy $\sigma_{\theta_i}$, let $\widehat{\sigma}_{\theta_i}: \mathcal{H} \rightarrow \Delta (A_1)$ be defined as:
\begin{equation}\label{B.3}
    \widehat{\sigma}_{\theta_i}(h^t) \equiv \left\{ \begin{array}{ll}
\overline{a}_1^{\sigma}(h^t) & \textrm{ if } h^t \in \mathcal{H}^{\sigma} \textrm{ and } \overline{a}_1^{\sigma}(h^t) \in \textrm{supp}\big(\sigma_{\theta_i}(h^t) \big)\\
\sigma_{\theta_i}(h^t) & \textrm{ otherwise}.
\end{array} \right.
\end{equation}
By construction, $\widehat{\sigma}_{\theta_i}$ is type $\theta_i$'s best reply against $\sigma_2$.
Let $\mathcal{H}^{(\widehat{\sigma}_{\theta_i},\sigma_2)}$ be the set of histories that occur with positive probability under
$(\widehat{\sigma}_{\theta_i},\sigma_2)$. Let
\begin{equation}\label{B.4}
    \overline{\mathcal{H}}^{\sigma,\theta_i} \equiv \Big\{
    h^t \in \overline{\mathcal{H}}^{\sigma}
    \Big| \theta_i \in \Theta^{\sigma}(h^t) \textrm{ and } \overline{a}_1^{\sigma}(h^t) \notin \textrm{supp}\big(\sigma_{\theta_i}(h^t) \big)
    \Big\}.
\end{equation}
Intuitively, $h^t \in \overline{\mathcal{H}}^{\sigma,\theta_i}$ if and only if
\begin{enumerate}
  \item At every $h^s \prec h^t$,
type $\theta_i$'s equilibrium strategy $\sigma_{\theta_i}$ plays $\overline{a}_1^{\sigma}(h^s)$ with positive probability. This comes from $h^t \in \overline{\mathcal{H}}^{\sigma}$ and $\theta_i \in \Theta^{\sigma}(h^t)$.
  \item At $h^t$,
type $\theta_i$ plays $\overline{a}_1^{\sigma}(h^t)$ with zero probability. This comes from $\overline{a}_1^{\sigma}(h^t) \notin \textrm{supp}\big(\sigma_{\theta_i}(h^t) \big)$.
\end{enumerate}
Consider type $\theta_i$'s payoff if he plays $\widehat{\sigma}_{\theta_i}$ and player $2$ plays $\sigma_2$. For any given $h^t \in \mathcal{H}^{(\widehat{\sigma}_{\theta_i},\sigma_2)}$,
\begin{enumerate}
  \item If there \textit{does not exist} $h^s \preceq h^t$ such that $h^s \in \overline{\mathcal{H}}^{\sigma,\theta_i}$, then  type $\theta_i$'s stage-game payoff at $h^t$ and at all histories preceding $h^t$ is no more than $1-\theta_i$.
  \item If there \textit{exists} $h^s \preceq h^t$ such that $h^s \in \overline{\mathcal{H}}^{\sigma,\theta_i}$, then I show below that type $\theta_i$'s continuation payoff at $h^s$ is no more than $1-\theta_i$.

      First, since $h^s \in \overline{\mathcal{H}}^{\sigma,\theta_i}$,
      after player $2$ observes $\overline{a}_1^{\sigma}(h^s)$ at $h^s$, $\theta_i$ is no longer in the support of player $2$'s posterior belief.
      Therefore, for every $h^{s+1} \succ h^s$ with $\overline{a}_1^{\sigma}(h^s)$ being played at $h^s$, there exist at most $n$ types in the support of player $2$'s posterior belief at $h^{s+1}$.

      Let $\theta_j$ be the lowest-cost type in the support of P2's belief at $h^{s+1}$.
According to the induction hypothesis, type $\theta_j$'s continuation payoff \textit{after} playing
      $\overline{a}_1^{\sigma}(h^s)$ at $h^s$ is no more than $1-\theta_j$. Type $\theta_j$'s stage-game payoff by playing $\overline{a}_1^{\sigma}(h^s)$ at $h^s$ is also no more than $1-\theta_j$. This implies that his continuation payoff at $h^s$ is at most $1-\theta_j$.

      Therefore, type $\theta_j$'s continuation payoff by deviating to $\widehat{\sigma}_{\theta_i}$ starting from $h^s$ is no more than $1-\theta_j$. Since $\theta_i< \theta_j$, and the maximal difference between type $\theta_i$ and $\theta_j$'s stage-game payoff is $\theta_j-\theta_i$, we know that type $\theta_i$'s continuation payoff at $h^s$ by playing $\widehat{\sigma}_{\theta_i}$ is no more than $1-\theta_i$.
\end{enumerate}
The two parts together imply that when $|\Theta^{\sigma}(h^0)|=n+1$.
\begin{enumerate}
  \item type $\theta_i$'s stage-game payoff before reaching any history that belong to $\overline{\mathcal{H}}^{\sigma,\theta_i}$ is no more than $1-\theta_i$,
  \item type $\theta_i$'s continuation payoff at any history belonging to $\overline{\mathcal{H}}^{\sigma,\theta_i}$ is no more than $1-\theta_i$.
\end{enumerate}
Therefore, $\theta_i$'s discounted average payoff in period $0$ is no more than $1-\theta_i$.
\end{proof}

\subsection{Necessity of Constraint (\ref{1.3})}
Suppose toward a contradiction that $(v_1,...,v_m)$ is an equilibrium payoff and there exists
$j \in \{1,2,...,m\}$ such that $v_j > v_j(\gamma^*)$. Then given the constraint established in the first part that $v_1 \leq 1-\theta_1$, we know that $j >1$.
Under the probability measure over $\mathcal{H}$ induced by $(\sigma_{\theta_j},\sigma_2)$, let
 $X^{(\sigma_{\theta_j},\sigma_2)}$ be the occupation measure of outcome $(T,H)$ and let $Y^{(\sigma_{\theta_j},\sigma_2)}$ be the occupation measure of outcome $(T,L)$. Since $v_j >v_j(\gamma^*)$, we have:
\begin{equation}\label{B.3}
    \frac{\displaystyle X^{(\sigma_{\theta_j},\sigma_2)} }{\displaystyle Y^{(\sigma_{\theta_j},\sigma_2)}} < \frac{\gamma^*}{1-\gamma^*}.
\end{equation}
Let the value of the left-hand-side be $\frac{\gamma}{1-\gamma}$ for some $\gamma \in [0, \gamma^*)$.

For every $h^{\tau} \in \mathcal{H}$,
let $\sigma_{\theta_j}(h^{\tau}) \in \Delta (A_1)$ be the (mixed) action prescribed by $\sigma_{\theta_j}$ at $h^{\tau}$ and let
$\alpha_1(\cdot|h^{\tau})$ be player $2$'s expected action of player $1$'s at $h^{\tau}$.
Let $d(\cdot \| \cdot)$ be the Kullback-Leibler divergence between two action distributions.
Suppose player $1$ plays according to $\sigma_{\theta_j}$, the result in Gossner (2011) implies that:
\begin{equation}\label{B.4}
  \mathbb{E}^{(\sigma_{\theta_j},\sigma_2)} \Big[   \sum_{\tau=0}^{+\infty} d(\sigma_{\theta_j}(h^{\tau})||\alpha_1(\cdot|h^{\tau}))
    \Big] \leq -\log \pi_0(\theta_j).
\end{equation}
This implies that for every $\epsilon >0 $,
the expected number of periods such that $ d(\sigma_{\theta_j}(h^{\tau})||\alpha_1(\cdot|h^{\tau})) > \epsilon$ is no more than
\begin{equation}\label{B.5}
    T(\epsilon) \equiv \Big\lceil \frac{-\log \pi_0(\theta_j)}{\epsilon}\Big\rceil.
\end{equation}
Let
\begin{equation}\label{B.6}
    \epsilon \equiv d\Big(
    \frac{\gamma+2\gamma^*}{3} H +(1- \frac{\gamma+2\gamma^*}{3}) L
    \Big\|
     \gamma^* H +(1- \gamma^*) L
    \Big),
\end{equation}
and let $\delta$ be large enough such that:
\begin{equation}\label{B.7}
      \frac{\displaystyle X^{(\sigma_{\theta_j},\sigma_2)} }{\displaystyle Y^{(\sigma_{\theta_j},\sigma_2)}-(1-\delta^{T(\epsilon)})} < \frac{2\gamma +\gamma^*}{3-2\gamma-\gamma^*}.
\end{equation}
According to (\ref{B.4}) and (\ref{B.5}), if type $\theta_j$ plays according to her equilibrium strategy, then there are at most $T(\epsilon)$ periods in which player $2$'s expectation over player $1$'s action differs from $\sigma_{\theta_j}$
by more than $\epsilon$. According to (\ref{B.6}), aside from $T(\epsilon)$ periods, player $2$ will trust player $1$ at $h^t$ only when $\sigma_{\theta_j}(h^t)$ assigns probability at least $\frac{\gamma + 2\gamma^*}{3}$ to $H$.
Therefore, under the probability measure induced by $(\sigma_{\theta_j},\sigma_2)$, the occupation measure with which player $2$ trusts player $1$ is at most:
\begin{equation}\label{B.8}
    \underbrace{(1-\delta^{T(\epsilon)})}_{\textrm{periods in which player 2's prediction is wrong}} +\underbrace{\Big(
    X^{(\sigma_{\theta_j},\sigma_2)}+Y^{(\sigma_{\theta_j},\sigma_2)}- (1-\delta^{T(\epsilon)})
    \Big)\frac{2\gamma +\gamma^*}{\gamma +2\gamma^*}}_{\textrm{maximal frequency with which player 2 trusts after he learns}},
\end{equation}
which is strictly less than $X^{(\sigma_{\theta_j},\sigma_2)}+Y^{(\sigma_{\theta_j},\sigma_2)}$ when $\delta$ is close enough to $1$. This leads to a contradiction.

\subsection{Proof of Lemma \ref{L4.1}}
Constraint (\ref{1.3}) implies that:
\begin{equation*}
    (1-\theta_1) \alpha(H) +\alpha(L) \leq (1-\theta_1) \alpha(H) + \frac{1-\gamma^*}{\gamma^*}\alpha(H) = (\frac{1}{\gamma^*}-\theta_1)\alpha(H),
\end{equation*}
or equivalently,
\begin{equation*}
    \alpha(H) \geq \frac{\gamma^* }{1-\gamma^* \theta_1} \Big( (1-\theta_1) \alpha(H) +\alpha(L) \Big).
\end{equation*}
Recall that $\theta_1$ is the lowest-cost type.
The objective function (\ref{1.1}) can be rewritten as:
\begin{equation*}
    (1-\theta_j) \alpha(H) +\alpha(L)
    =(1-\theta_1) \alpha(H) +\alpha(L)-\underbrace{(\theta_j-\theta_1)}_{\geq 0} \underbrace{\alpha(H)}_{\geq \frac{\gamma^* }{1-\gamma^* \theta_1} \Big( (1-\theta_1) \alpha(H) +\alpha(L) \Big)}
\end{equation*}
\begin{equation*}
    \leq \Big( 1-(\theta_j-\theta_1) \frac{\gamma^*}{1-\gamma^* \theta_1} \Big) \Big( (1-\theta_1) \alpha(H) +\alpha(L) \Big)
    =\frac{1-\gamma^* \theta_j}{1-\gamma^* \theta_1} \Big( (1-\theta_1) \alpha(H) +\alpha(L) \Big) .
\end{equation*}
According to constraint (\ref{1.2}) that $(1-\theta_1) \alpha(H) +\alpha(L)  \leq 1-\theta_1$, we have the following upper bound of the objective function (\ref{1.1}):
\begin{equation*}
(1-\theta_j) \alpha(H) +\alpha(L)\leq
\frac{1-\gamma^* \theta_j}{1-\gamma^* \theta_1} \Big( (1-\theta_1) \alpha(H) +\alpha(L) \Big)  \leq   (1-\gamma^* \theta_j ) \frac{1-\theta_1}{1-\gamma^* \theta_1}=v_j^*.
\end{equation*}
The above upper bound is attained by the following distribution over action profiles:
\begin{equation*}
    \alpha(H)=\frac{(1-\theta_1) \gamma^*}{1-\gamma^* \theta_1}, \alpha(L)=\frac{(1-\theta_1) (1-\gamma^*)}{1-\gamma^* \theta_1}, \textrm{ and } \alpha(N)=\frac{\theta_1 (1-\gamma^*)}{1-\gamma^* \theta_1},
\end{equation*}
which satisfies constraints (\ref{1.2}) and (\ref{1.3}). Therefore, the value of the optimization problem is $v_j^*$.

\section{Results on Equilibrium Behavior}\label{secD}
In section \ref{secC.1}, I show Theorem 3. In section \ref{secC.2}, I allow player $1$ to have arbitrary payoff functions, and show that in any equilibrium in which a \textit{normal type} approximately attains his Stackelberg payoff, any payoff type who plays the Stackelberg action in every period has to be indifferent between all outcomes in the one-shot game.
\subsection{Proof of Theorem 3}\label{secC.1}
For the first statement, suppose there exists type $\theta_i \neq \theta_m$ and
a pure strategy $\widehat{\sigma}_{\theta_i}$ that is type $\theta_i$'s best reply to $\sigma_2$, such that
\begin{equation}\label{C.3}
    \frac{\mathbb{E}^{(\widehat{\sigma}_{\theta_i},\sigma_2)} \Big[\sum_{t=0}^{\infty} (1-\delta)\delta^t \mathbf{1}\{y_t=H\}\Big] }{ \mathbb{E}^{(\widehat{\sigma}_{\theta_i},\sigma_2)} \Big[\sum_{t=0}^{\infty} (1-\delta)\delta^t \mathbf{1}\{y_t=L\}\Big]}
    = \frac{\gamma_i}{1-\gamma_i}
\end{equation}
for some $\gamma_i < \gamma^*$. Let $p_i$ be the discounted average frequency with which player $2$ plays $T$ under $(\widehat{\sigma}_{\theta_i},\sigma_2)$.

Let $\widehat{\sigma}_{\theta_m}$ be an arbitrary pure-strategy best reply of type $\theta_m$ against $\sigma_2$.
Let $p_m$ be the discounted average frequency with which player $2$ plays $T$ under $(\widehat{\sigma}_{\theta_m},\sigma_2)$ and let $\gamma_m$ be pinned down via:
\begin{equation}\label{C.4}
    \frac{\mathbb{E}^{(\widehat{\sigma}_{\theta_m},\sigma_2)} \Big[\sum_{t=0}^{\infty} (1-\delta)\delta^t \mathbf{1}\{y_t=H\}\Big] }{ \mathbb{E}^{(\widehat{\sigma}_{\theta_m},\sigma_2)} \Big[\sum_{t=0}^{\infty} (1-\delta)\delta^t \mathbf{1}\{y_t=L\}\Big]}
   = \frac{\gamma_m}{1-\gamma_m}.
\end{equation}
The long-run player's ex ante incentive constraints, namely, first,
type $\theta_i$ prefers $\widehat{\sigma}_{\theta_i}$ to $\widehat{\sigma}_{\theta_m}$,
and second,
type $\theta_m$ prefers $\widehat{\sigma}_{\theta_m}$ to $\widehat{\sigma}_{\theta_i}$
imply that
$p_i \geq p_m$ and $\gamma_i \geq \gamma_m$. This further implies that according to type $\theta_m$'s equilibrium strategy $\sigma_{\theta_m}$,
\begin{equation*}
    \frac{\mathbb{E}^{(\sigma_{\theta_m},\sigma_2)} \Big[\sum_{t=0}^{\infty} (1-\delta)\delta^t \mathbf{1}\{y_t=H\}\Big] }{ \mathbb{E}^{(\sigma_{\theta_m},\sigma_2)} \Big[\sum_{t=0}^{\infty} (1-\delta)\delta^t \mathbf{1}\{y_t=L\}\Big]}
  \leq \frac{\gamma_i}{1-\gamma_i},
\end{equation*}
or equivalently,
\begin{equation}\label{C.5}
    \gamma_i \mathbb{E}^{(\sigma_{\theta_m},\sigma_2)} \Big[\sum_{t=0}^{\infty} (1-\delta)\delta^t \mathbf{1}\{y_t=L\}\Big]-   (1-\gamma_i) \mathbb{E}^{(\sigma_{\theta_m},\sigma_2)} \Big[\sum_{t=0}^{\infty} (1-\delta)\delta^t \mathbf{1}\{y_t=H\}\Big] >0.
\end{equation}
Since type $\theta_m$'s payoff from $\sigma_{\theta_m}$ is at least $v_m^*-\varepsilon$, which is strictly greater than $1-\theta_m$ when $\varepsilon$ is small enough. This places a lower bound on $p_m$. If type $\theta_m$ plays according to $\sigma_{\theta_m}$, then the learning arguments in Fudenberg and Levine (1992) and Gossner (2011) imply that for every $\varepsilon>0$, there exists $\overline{\delta}$ such that when $\delta > \overline{\delta}$,
\begin{equation}\label{C.6}
 \gamma^* \mathbb{E}^{(\sigma_{\theta_m},\sigma_2)} \Big[\sum_{t=0}^{\infty} (1-\delta)\delta^t \mathbf{1}\{y_t=L\}\Big]-   (1-\gamma^*) \mathbb{E}^{(\sigma_{\theta_m},\sigma_2)} \Big[\sum_{t=0}^{\infty} (1-\delta)\delta^t \mathbf{1}\{y_t=H\}\Big]
 < \varepsilon.
\end{equation}
This contradicts (\ref{C.5}) once we pick $\varepsilon$ to be small enough, which establishes the lower bound on the relative frequencies of actions.

For the second statement, suppose towards a contradiction that according to one of type $\theta$ $(\neq \theta_1)$'s
pure-strategy best reply to $\sigma_2$, denoted by $\widehat{\sigma}_{\theta}$,
\begin{equation}\label{C.7}
     \frac{\mathbb{E}^{(\widehat{\sigma}_{\theta},\sigma_2)} \Big[\sum_{t=0}^{\infty} (1-\delta)\delta^t \mathbf{1}\{y_t=H\}\Big] }{ \mathbb{E}^{(\widehat{\sigma}_{\theta},\sigma_2)} \Big[\sum_{t=0}^{\infty} (1-\delta)\delta^t \mathbf{1}\{y_t=L\}\Big]}
    =\frac{\gamma}{1-\gamma}
\end{equation}
where $\gamma > \gamma^*$. Let
\begin{equation*}
    p \equiv \mathbb{E}^{(\widehat{\sigma}_{\theta},\sigma_2)} \Big[\sum_{t=0}^{\infty} (1-\delta)\delta^t \mathbf{1}\{y_t=H\}\Big]
    +\mathbb{E}^{(\widehat{\sigma}_{\theta},\sigma_2)} \Big[\sum_{t=0}^{\infty} (1-\delta)\delta^t \mathbf{1}\{y_t=L\}\Big].
\end{equation*}
If type $\theta_1$ plays according to $\widehat{\sigma}_{\theta}$, her payoff is $p(1-\gamma \theta_1)$. According to Theorem 1,
\begin{equation}\label{C.8}
    p(1-\gamma \theta_1) \leq 1-\theta_1.
\end{equation}
If type $\theta$ plays according to $\widehat{\sigma}_{\theta}$, she receives her equilibrium payoff, which is $p(1-\gamma \theta)$. The equilibrium payoff is within $\varepsilon$ of $v^*$ implies that:
\begin{equation}\label{C.9}
    p(1-\gamma \theta) \geq \frac{1-\theta_1}{1-\gamma^* \theta_1} (1-\gamma^* \theta) -\varepsilon.
\end{equation}
Inequalities (\ref{C.8}) and (\ref{C.9}) together imply that:
\begin{equation}\label{C.10}
    \varepsilon > (1-\theta_1) \Big\{
    \frac{1-\gamma^* \theta}{1-\gamma^* \theta_1} -\frac{1-\gamma \theta}{1-\gamma \theta_1}
    \Big\}.
\end{equation}
The RHS is strictly positive since $\gamma > \gamma^*$ and $\theta > \theta_1$. As a result, inequality (\ref{C.10}) cannot hold for $\varepsilon$ smaller than the RHS. For every $\gamma > \gamma^*$, take $\varepsilon$ to be smaller than
\begin{equation*}
    \min_{\theta \neq \theta_1} (1-\theta_1) \Big\{
    \frac{1-\gamma^* \theta}{1-\gamma^* \theta_1} -\frac{1-\gamma \theta}{1-\gamma \theta_1}
    \Big\},
\end{equation*}
we obtain a contradiction. This establishes the upper bound on the relative frequencies.
\subsection{Payoff Types who Behave Like Stackelberg Commitment Type}\label{secC.2}
I explore when can a rational-type patient player mixes between $H$ and $L$ at every history in the trust game, i.e., behaves like the Stackelberg commitment type in canonical reputation models. Players' stage-game payoffs are shown in Figure 5. Player $1$'s stage-game payoff is a function of her type $\theta \in \Theta$, where $\Theta$ is finite.
Different from the baseline model which requires all types of player $1$ sharing the same ordinal preferences over stage-game outcomes,
the current setup allows for arbitrary preferences for player $1$, and only requires one type of player $1$ having the ordinal preference in the baseline model, which I call the \textit{normal type}.
\begin{figure}
\begin{center}
\begin{tikzpicture}[scale=0.22]
\draw [->, thick] (0,10)--(-5,5);
\draw [->, thick] (0,10)--(5,5);
\draw [->, thick] (-5,5)--(-10,0);
\draw [->, thick] (-5,5)--(0,0);
\draw (0,10)--(-2,8)node[left]{T};
\draw (0,10)--(2,8)node[right]{N};
\draw (-5,5)--(-7,3)node[left]{H};
\draw (-5,5)--(-3,3)node[right]{L};
\draw [ultra thick] (0,9.9)--(0,10.1)node[above, blue]{P2};
\draw [ultra thick] (-5.1,5)--(-4.9,5)node[left, red]{P1};
\draw [ultra thick] (5,5)--(5,4.9)node[below]{({\color{red}{$z(\theta)$}}, {\color{blue}{$0$}})};
\draw [ultra thick] (-10,0)--(-10,-0.1)node[below]{({\color{red}{$x(\theta)$}}, {\color{blue}{$b$}})};
\draw [ultra thick] (0,0)--(0,-0.1)node[below]{({\color{red}{$y(\theta)$}}, {\color{blue}{$-c$}})};
\end{tikzpicture}
\caption{Generalized Stage Game}
\end{center}
\end{figure}
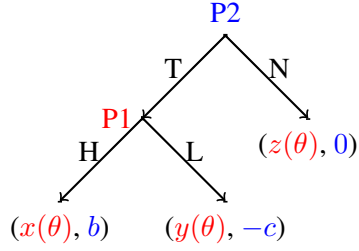
\begin{Proposition}\label{PropC.1}
Suppose there exists $\theta^* \in \Theta$, such that $y(\theta^*)>x(\theta^*)>z(\theta^*)$.
For every small enough $\varepsilon>0$, there exists $\underline{\delta} \in (0,1)$, such that when $\delta >\underline{\delta}$, if there exists BNE such that:
\begin{itemize}
  \item type $\theta^*$ attains within $\varepsilon$ of her Stackelberg payoff,
  \item there exists $\theta' \in \Theta$ such that type $\theta'$ mixes between $H$ and $L$ at every history,
\end{itemize}
then it must be the case that $x(\theta')=y(\theta')=z(\theta')$.
\end{Proposition}
\begin{proof}
Suppose there exists such an equilibrium $\sigma$. For every $\varepsilon>0$, there exists $\underline{\delta}$ such that when $\delta>\underline{\delta}$, such that
\begin{equation}
    \frac{\mathbb{E}^{(\sigma_{\theta^*},\sigma_2)} \Big[\sum_{t=0}^{\infty} (1-\delta)\delta^t \mathbf{1}\{y_t=H\}\Big]}{ \mathbb{E}^{(\sigma_{\theta^*},\sigma_2)} \Big[\sum_{t=0}^{\infty} (1-\delta)\delta^t \mathbf{1}\{y_t=L\}\Big]}
    \geq \frac{\gamma^*-\varepsilon}{1-\gamma^*+\varepsilon}.
\end{equation}
Since type $\theta^*$'s payoff is within $\varepsilon$ of her Stackelberg payoff $\gamma^* x(\theta^*)+(1-\gamma^*) y(\theta^*)$, we have:
\begin{equation}\label{D.2}
    \mathbb{E}^{(\sigma_{\theta^*},\sigma_2)} \Big[\sum_{t=0}^{\infty} (1-\delta)\delta^t \mathbf{1}\{y_t=N\}\Big] \leq
    \underbrace{\frac{1+y(\theta^*)-x(\theta^*)}{(1-\gamma^*+\varepsilon) y(\theta^*) +(\gamma^*-\varepsilon) x(\theta^*)-z(\theta^*)}}_{\equiv C}
\varepsilon    .
\end{equation}
For type $\theta'$ to mix at every history, she is indifferent between $\sigma_{\theta^*}$, the strategy of playing $H$ at every history, denoted by $\sigma_H$, and the strategy of playing $L$ at every history, denoted by $\sigma_L$. Inequality (\ref{D.2}) implies that unless $x(\theta')=y(\theta')$, type $\theta'$ receives equilibrium payoff \textit{strictly} between $x(\theta')$ and $y(\theta')$.
If type $\theta'$ plays $\sigma_H$, then her payoff is between $x(\theta')$ and $z(\theta')$. If type $\theta'$ plays $\sigma_L$, then her payoff is between $y(\theta')$ and $z(\theta')$, and it cannot be equal $y(\theta')$.
These three payoffs coincide, which I denote by $\Pi(\theta')$.
I consider two cases separately.

First, suppose $x(\theta') = y(\theta') \neq z(\theta')$. Then $\Pi(\theta')$ is a convex combination of $y(\theta')$ and $z(\theta')$, with the convex weight on $z(\theta')$ being less than the RHS of (\ref{D.2}). Type $\theta'$ being indifferent between $\sigma_{\theta^*}$ and $\sigma_L$ implies that type $\theta^*$'s payoff by playing $\sigma_L$ is at least:
\begin{equation*}
   (1-C\varepsilon) y(\theta^*) +C\varepsilon z(\theta^*),
\end{equation*}
which is strictly greater than his Stackelberg payoff when $\varepsilon$ is small enough. This implies that type $\theta^*$ has a strict incentive to deviate to $\sigma_L$, which leads to a contradiction.

Next, suppose $x(\theta') \neq y(\theta')$. Consider three subcases. First, if $z(\theta') \geq \max\{y(\theta'),x(\theta')\}$, then any convex combination between $z(\theta')$ and the larger one among $x(\theta')$ and $y(\theta')$ is strictly larger than a number strictly between $x(\theta')$ and $y(\theta')$. This leads to a contradiction.
Similarly, if $z(\theta') \leq \min \{y(\theta'),x(\theta')\}$, then  any convex combination between $z(\theta')$ and the smaller one among $x(\theta')$ and $y(\theta')$ is strictly smaller than a number strictly between $x(\theta')$ and $y(\theta')$. This leads to a contradiction.
Third, if $z(\theta')$ is strictly between $x(\theta')$ and $y(\theta')$. Then for a convex combination between $z(\theta')$ and $x(\theta')$ to equal a convex combination between $z(\theta')$ and $y(\theta')$, both of which attach convex weight $1$ to $z(\theta')$, i.e., unless player $1$ has played different actions in the past, otherwise, player $2$ never plays $T$. On the other hand, if player $1$ plays according to $\sigma_{\theta^*}$, then the discounted average frequency of outcome $N$ is close to $0$. Hence, there exists $h^t$ that occurs with positive probability under $(\sigma_{\theta^*},\sigma_2)$, such that $T$ is played with positive probability under $h^t$ but not at any predecessor of $h^t$. This leads to a contradiction since by using strategy $\sigma_H$, player $1$ can also reach history $h^t$, which implies that the discounted average frequency of $N$ under $(\sigma_H,\sigma_2)$ is strictly smaller than $1$.
\end{proof}

\section{Generalizations \& Robustness}\label{secA}
I state generalizations of Theorems \ref{Theorem3.1} and \ref{Theorem3.2} to a class of games.
I setup the model and state Theorems 1' and 2' when players move \textit{simultaneously} in the stage game. Analogous results also hold in sequential-move stage games. Due to length restrictions, the proofs of Theorems 1' and 2' are
available upon request.


Consider a repeated game in discrete time between an informed player $1$ with discount factor $\delta_1 \in (0,1)$, and
an uninformed player $2$ with discount factor $\delta_2 \in [0,1)$.
\begin{itemize}
  \item Player $1$ has a perfectly persistent type $\theta \in \Theta$. Player $2$'s full support prior is $\pi \in \Delta (\Theta)$.
  \item Player $1$'s action $a_1 \in A_1$, player $2$'s action $a_2 \in A_2$. A pure action profile is $a \in A\equiv A_1 \times A_2$, and
 player $1$'s mixed action is denoted by $\alpha_1 \in \Delta(A_1)$.
  \item Players' stage game payoffs are $u_1(\theta,a_1,a_2)$ and $u_2(a_1,a_2)$, i.e., values are private.
  \item Player 2's stage-game best reply to $\alpha_1 \in \Delta (A_1)$ is $\textrm{BR}_2(\alpha_1)$, which is a nonempty subset of $A_2$.
  \item For every $\theta \in \Theta$, the set of type $\theta$'s pure Stackelberg action is $\arg\max_{a_1 \in A_1} \min_{a_2 \in \textrm{BR}_2(a_1)}u_1(\theta,a_1,a_2)$.
  \item Players' pure actions are perfectly monitored. Players \textit{cannot} observe each other's mixed actions.
  \item I assume that $\Theta$ and $A_1$ are finite sets, and $|A_2|=2$.
\end{itemize}
I start from Assumption \ref{Ass1} that is satisfied for generic $u_1$ and $u_2$:
\begin{Assumption}\label{Ass1}
Players' stage-game payoff functions $u_1$ and $u_2$ satisfy:
\begin{enumerate}
  \item For every pure action $a_1 \in A_1$, $\textrm{BR}_2(a_1)$ is a singleton.
  \item For every $\theta \in \Theta$, type $\theta$ has a unique pure Stackelberg action.
\end{enumerate}
\end{Assumption}
Assumption \ref{Ass2} is called \textit{monotone-supermodularity} (MSM), which captures the lack-of-commitment problem.
\begin{Assumption}[MSM]\label{Ass2}
$\Theta$ and $A_2$ are fully ordered sets, and $A_1$ is a lattice, such that:
\begin{itemize}
\item[1.] $u_1(\theta,a_1,a_2)$ is strictly decreasing in $a_1$, and is strictly increasing in $a_2$.
\item[2.] $u_1(\theta,a_1,a_2)$ has strictly increasing differences in $\theta$ and $a_1$,\\
$u_1(\theta,a_1,a_2)$ has weakly increasing differences in $\theta$ and $a_2$.
\item[3.] $u_2(a_1,a_2)$ has strictly increasing differences in $a_1$ and $a_2$.
\end{itemize}
\end{Assumption}
To map Assumption \ref{Ass2} into the seller-buyer application, let $a_1$ be the quality of good the seller supplies, $a_2$ represents whether the buyer makes the purchase or not, or whether she purchases the customized or standardized version, and $\theta$ measures the efficiency of the seller's production technology.
Assumption \ref{Ass2} requires that:
\begin{enumerate}
  \item supply high quality is strictly costly for the seller, but he strictly benefits from buyers' trusting action;
  \item a higher type faces lower cost to supply high quality, and values buyers' trust weakly more;
  \item a buyer has stronger incentive to trust when her expectation of product quality is higher.
\end{enumerate}
That being said, my general framework allows the seller to have
\begin{enumerate}
  \item any finite number of rational types,
  \item any finite number of actions, and the seller's effort can be multi-dimensional.
\end{enumerate}
My framework surpasses the generality of Mailath and Samuelson (2001), Phelan (2006), Ekmekci (2011), Liu (2011), and Jehiel and Samuelson (2012) that focus on $2 \times 2$ stage-games with player $1$ having only one rational type.
It also allows for non-separability between $\theta$ and $a_2$, and between $a_1$ and $a_2$ in seller's payoff,
i.e., my framework can accommodate, but is not limited to the case in
which player $1$'s benefit from player $2$s' trust being independent of her type and her action.

Let $\Theta \equiv \{\theta_1,\theta_2,...,\theta_m\}$
with $\theta_1 \succ \theta_2 \succ ... \succ \theta_m$.
For $i \in \{1,2\}$,
let $\overline{a}_i \equiv \max A_i$ and $\underline{a}_i \equiv \min A_i$.
Under Assumption \ref{Ass2}, $(\underline{a}_1,\underline{a}_2)$ is all types of player $1$'s minmax outcome.

Let $a_1^*(\theta) \in A_1$ be
type $\theta$'s \textit{pure Stackelberg action}, namely, her optimal commitment \textit{if she can only commit to pure actions}.
This is uniquely defined under the second statement of Assumption \ref{Ass1}. I assume that the most efficient type of seller finds it optimal to commit to supply the highest quality.
\begin{Assumption}\label{Ass3}
$a_1^*(\theta_1)=\overline{a}_1$.
\end{Assumption}
Assumption \ref{Ass3} allows some types of the seller to be \textit{inefficient}, i.e, their cost of supplying high quality is so high that they
strictly prefer the minmax outcome $(\underline{a}_1,\underline{a}_2)$ to the highest outcome $(\overline{a}_1,\overline{a}_2)$.

I generalize Theorem 1 by characterizing every type of P1's \textit{highest equilibrium payoff} when $\delta_1$ is close to $1$ and $\delta_2$ is close to or equal to $0$.
For every $j \in \{1,2,...,m\}$, let $v_j^*$ be the value of the following problem:
\begin{equation}\label{3.1}
  \max_{\alpha \in \Delta (A_1 \times A_2)} \sum_{a \in A} \alpha(a) u_1(\theta_j,a)
\end{equation}
subject to:
\begin{equation}\label{3.2}
    \sum_{a \in A} \alpha(a) u_1(\theta_1,a) \leq u_1(\theta_1,\overline{a}_1,\overline{a}_2),
\end{equation}
and for every $a_2^* \in A_2$ such that the marginal distribution of $\alpha$ on $A_2$ attaches positive probability to $a_2^*$,
\begin{equation}\label{3.3}
    a_2^* \in \arg\max_{a_2 \in A_2} \sum_{a_1 \in A_1} \alpha_{1} (a_1|a_2^*) u_2(a_1,a_2),
\end{equation}
with $\alpha_{1}(\cdot|a_2^*) \in \Delta (A_1)$ the distribution over player $1$'s actions conditional on $a_2^*$ induced by joint distribution $\alpha$.
\begin{Theorem1}
If players' stage-game payoffs satisfy Assumptions 1, 2 and 3, then for every $\varepsilon>0$, there exist $\underline{\delta}_1 \in (0,1)$ and
$\overline{\delta}_2 \in (0,1)$
such that for every $\delta_1 \in( \underline{\delta}_1, 1)$ and $\delta_2 \in [0,\overline{\delta}_2)$,
\begin{enumerate}
\item There exists no BNE in which type $\theta_1$'s payoff is strictly more than $v_1^*$. There exists no BNE in which type $\theta_j$'s payoff is strictly more than $v_j^*+\varepsilon$ for some $j \in \{2,3,...,m\}$.
\item There exists sequential equilibrium in which P1 attains payoff within $\varepsilon$ of $v^* \equiv (v_1^*,...v_m^*)$.
\end{enumerate}
\end{Theorem1}
The proof uses similar ideas as that of Theorem 1, and is available upon request.
According to Theorem 1', $v_j^*$ is type $\theta_j$'s highest equilibrium payoff.
Similar to the baseline model, the most efficient type $\theta_1$
cannot strictly benefit from incomplete information and her maximal payoff in the repeated incomplete information game coincides with her highest equilibrium payoff in the repeated complete information game. Second, all types except for type $\theta_1$ can strictly benefit from incomplete information. Moreover, every type's highest equilibrium payoff depends only on their own type and the most efficient type in the support of player $2$s' prior belief. Third,
 $v_j^*$ is pinned down by two constraints, which generalize constraints (\ref{1.2}) and (\ref{1.3}) in the baseline model. Constraint (\ref{3.2}) says that the most efficient type of P1 receives payoff no more than $u_1(\theta_1,\overline{a}_1,\overline{a}_2)$ under the distribution over stage-game outcome induced by type $\theta_j$. Similar to the baseline model,
this arises due to type $\theta_1$'s incentive constraint and that he cannot benefit from private information. This constraint is absent in commitment-type models as commitment types face no incentive constraint. Constraint (\ref{3.3}) says that conditional on every $a_2$ that occurs with positive probability under $\alpha$,
P2 has an incentive to play
$a_2$ against the conditional distribution over player $1$'s actions. This comes from P2's learning, namely, her prediction about P1's action is arbitrarily close to P1's action in the true state in all except for a bounded number of periods.

Let $\Theta^*$ be the set of types whose \textit{mixed Stackelberg payoffs} are strictly greater than their minmax payoffs:
\begin{equation}\label{3.10}
    \Theta^* \equiv \Big\{
    \theta
    \Big|
    \textrm{ there exist } \alpha_1 \in \Delta (A_1) \textrm{ and } a_2 \in \textrm{BR}_2(\alpha_1) \textrm{ such that }
    u_1(\theta,\alpha_1,a_2)> u_1(\theta, \underline{a}_1,\underline{a}_2)
    \Big\},
\end{equation}
Given $\sigma_2 :\mathcal{H} \rightarrow \Delta (A_2)$,
I say that $\sigma_{\theta}: \mathcal{H} \rightarrow \Delta(A_1)$ is \textit{stationary} with respect to $\sigma_2$
if it takes the same value for every $h^t$ that occurs with positive probability under $(\sigma_{\theta},\sigma_2)$.
I say that $\sigma_{\theta}: \mathcal{H} \rightarrow \Delta(A_1)$ is \textit{completely mixed} with respect to $\sigma_2$
if $\sigma_{\theta}(h^t)$ has full support for every $h^t$ that occurs with positive probability under $(\sigma_{\theta},\sigma_2)$.
\begin{Theorem2}
If $|\Theta^*| \geq 2$ and stage-game payoffs satisfy Assumptions \ref{Ass1}, \ref{Ass2}, and \ref{Ass3},
then
for every small enough $\varepsilon >0$, there exist $\underline{\delta}_1 \in (0,1)$ and $\overline{\delta}_2 \in (0,1)$ such that when $\delta_1 \in ( \underline{\delta}_1,1)$ and $\delta_2 < [0,\overline{\delta}_2)$,
in any BNE $\sigma \equiv \big((\sigma_{\theta})_{\theta \in \Theta}, \sigma_2 \big)$ that attains payoff within $\varepsilon$ of $v^*$:
\begin{enumerate}
  \item for every $\theta \in \Theta^*$ and every $\widehat{\sigma}_{\theta}$ that is type $\theta$'s best reply against $\sigma_2$, $\widehat{\sigma}_{\theta}$ is not completely mixed.
  \item  for every $\theta \in \Theta^*$, $\sigma_{\theta}$ is not stationary.
\end{enumerate}
\end{Theorem2}
Similar to that of Theorem \ref{Theorem3.2}, the proof of Theorem 2' uses the following ideas. First,
for any given type of player $1$, she \textit{cannot} strictly benefit from incomplete information after she becomes the most efficient type in the support of player $2$'s posterior belief. Second, if stage-game payoffs are \textit{monotone-supermodular} (Assumption 2) and player $2$'s stage-game action choice is binary, then in every BNE of this repeated signaling game:
\begin{itemize}
  \item if a less efficient type finds it weakly optimal to play his highest action $\overline{a}_1$ in every period,
then a more efficient type plays $\overline{a}_1$ with probability $1$ at every on-path history.
\item if a more efficient type finds it weakly optimal to play his lowest action $\underline{a}_1$ in every period,
then a less efficient type plays $\underline{a}_1$ with probability $1$ at every on-path history.
\end{itemize}
In general, this step uses a result on 1-shot signaling games developed in Liu and Pei (2020), that provides sufficient conditions on the sender's and the receiver's payoff functions, under which the sender's action is \textit{non-decreasing} with respect to her type in all equilibria.

Next, I use the following learning argument in Gossner (2011). Suppose $\theta$ is player $1$'s true type and $\sigma_{\theta}: \mathcal{H} \rightarrow \Delta (A_1)$ is type $\theta$'s equilibrium strategy, the expected number of periods in which player $2$s fail to play an $\varepsilon$-stage-game best reply against $\sigma_{\theta}(h^t)$ is uniformly bounded from above. When $\delta \rightarrow 1$, type $\theta$'s payoff from playing her equilibrium strategy cannot exceed her payoff from \textit{pre-committing} to strategy $\sigma_{\theta}$.
\end{spacing}

\end{document}